\documentclass[10pt]{article}
\usepackage{amsfonts}
\usepackage{amssymb}
\usepackage{amsthm}
\usepackage{amsmath}    
\usepackage{hyperref}
\usepackage{bbm}
\usepackage{graphicx}
\usepackage{latexsym}
\usepackage{mathrsfs}

\usepackage{geometry}
\DeclareMathAlphabet\EuFrak{U}{euf}{m}{n}	
\SetMathAlphabet\EuFrak{bold}{U}{euf}{b}{n}	

\newcommand{\cA}{\mathcal{A}}
\newcommand{\cB}{\mathcal{B}}

\newcommand{\cD}{\mathcal{D}}
\newcommand{\cE}{\mathcal{E}}

\newcommand{\cG}{\mathcal{G}}
\newcommand{\cH}{\mathcal{H}}

\newcommand{\cP}{\mathcal{P}}
\newcommand{\cQ}{\mathcal{Q}}

\newcommand{\cS}{\mathcal{S}}

\newcommand{\cU}{\mathcal{U}}

\newcommand{\rC}{\mathrm{C}}

\newcommand{\rF}{\mathrm{F}}

\newcommand{\rL}{\mathrm{L}}

\newcommand{\rd}{\mathrm{d}}

\newcommand{\rg}{\mathrm{g}}
\newcommand{\rh}{\mathrm{h}}

\newcommand{\rr}{\mathrm{r}}

\newcommand{\rt}{\mathrm{t}}
\newcommand{\ru}{\mathrm{u}}

\newcommand{\rw}{\mathrm{w}}

\newcommand{\sL}{\mathscr{L}}

\newcommand{\sP}{\mathscr{P}}

\newcommand{\Si}{\Sigma} 

\newcommand{\bN}{\mathbb{N}}
\newcommand{\bR}{\mathbb{R}}

\newcommand{\bZ}{\mathbb{Z}}

\newcommand{\si}{\sigma}

\newcommand{\eps}{\varepsilon}

\newcommand{\supp}{\mathrm{supp}}

\newcommand{\ad}{\mathrm{ad}}

\title{{\bf QED Representation for the Net of Causal Loops}} 

\author{
\textsc{Fabio Ciolli}$^{a*}$ \ \   \textsc{Giuseppe Ruzzi}$^{a}$\footnote{F.C. and G.R. are supported by PRIN-2010/2011.
F.C. is supported by the ERC Advanced Grant 227458 OACFT ``Operator Algebras and Conformal Field Theory".
} \ \    \textsc{Ezio Vasselli}$^{b}$\\[5pt]
\small{$^{a}$  Dipartimento di Matematica, Universit\`a di Roma ``Tor Vergata'',}\\
\small{Via della Ricerca Scientifica 1, I-00133 Roma,  Italy.}  \\
\small{\texttt{ciolli@mat.uniroma2.it}, \texttt{ruzzi@mat.uniroma2.it}} \\[5pt]
\small{$^{b}$Dipartimento di Matematica, Universit\`a di Roma ``La Sapienza'',}\\
\small{Piazzale Aldo Moro 5, I-00185 Roma, Italy.} \\
\small{\texttt{ezio.vasselli@gmail.com}}\\\\
\emph{Dedicated to Roberto Longo on the occasion of his sixtieth birthday}
}
\date{}

\begin{document}
\maketitle

\begin{abstract}
The present work tackles the existence of \emph{local} gauge symmetries in the setting  of Algebraic Quantum Field Theory (AQFT). 
The net of causal loops, previously introduced by the authors,  is a  model independent construction of a covariant net of local $\rC^*$-algebras on any 4-dimensional globally hyperbolic space-time, aimed to capture some structural properties of any reasonable quantum gauge theory. Representations of this net can be  described by causal and covariant connection systems, and local gauge transformations arise as maps between equivalent  connection systems.  The present paper completes these abstract results,  realizing QED as a representation of the net of causal loops in Minkowski 
space-time. More precisely,  we map the quantum electromagnetic field $F_{\mu\nu}$, \emph{not free in general}, into 
a representation of the net of causal loops and show that the corresponding connection system and local gauge transformations
find a counterpart in terms of $F_{\mu\nu}$. 
\end{abstract}

\newpage

\tableofcontents
\markboth{Contents}{Contents}


  \theoremstyle{plain}
  \newtheorem{definition}{Definition}[section]
  \newtheorem{theorem}[definition]{Theorem}
  \newtheorem{proposition}[definition]{Proposition}
  \newtheorem{corollary}[definition]{Corollary}
  \newtheorem{lemma}[definition]{Lemma}

  \theoremstyle{definition}
  \newtheorem{remark}[definition]{Remark}
    \newtheorem{example}[definition]{Example}

\theoremstyle{definition}
  \newtheorem{ass}{\underline{\textit{Assumption}}}[section]


\numberwithin{equation}{section}

\section{Introduction} 
The Standard Model of elementary particles is a successful physical theory tested by awesome experiments, 
although its constituents like QED or QCD are still out of a rigorous mathematical comprehension.
Yet, similar situations hold for any quantum gauge theory and for the attempts to generalize the achievements of QFT to the realm of general relativity and gravity.  Only some instances, like in  quantum free fields models or  in low dimensional space-time, have  a rigorous  description.
On the other hand, more insights may arrive from an axiomatic approach to the theory of quantum fields.\\ 
\indent
In the present paper we interest in local gauge theory through AQFT, the axiomatic approach describing local observables in terms of nets of abstract $\rC^*$-algebras undergoing few basic physical requirements, see the reference book by R. Haag \cite{Haa}.
Gauge theories have been investigated in this framework mainly in terms of superselection sectors (equivalent classes of representations), the so called DHR analysis.
In a 4-dimensional Minkowski space-time \cite{DHR}, this is a real mathematical-physics outstanding result, grasping the r\^ole played by the DHR-sectors as the dual category of the group of \emph{global} gauge symmetries \cite{DR90}. 
This result disseminated over other cases of superselection theory, as in low-dimensional CFT, e.g.\ \cite{Ka04}, in curved space-time \cite{GLRV, BR}, also in combination with  other non-commutative geometry features \cite{CHL,CM}.
Moreover, authors investigated other aspects of QFT using the language of AQFT, e.g.\ noncommutative space-time \cite{DFR}, the AdS/CFT correspondence \cite{Re00}, locally covariant QFT \cite{BFV}, perturbative QFT and the renormalization group \cite{BDF, FRjlec, FRj}.\\ 
\indent It was clear from the beginning that DHR-sectors  were not tailored for describing charges of electromagnetic type and also,  as pointed out later, for the charges of purely massive theories. The right sectors for the massive theory case were introduced by Buchholz and Fredenhagen \cite{BF}.
Instead, suitable requirements for charges of electromagnetic type have only recently been introduced by  Buchholz and Roberts\footnote{Actually, a long-standing project started by these authors and S. Doplicher.}: for representations in theories with long range forces, like QED, they used the notion of \emph{charge class}, restricting the states of interest to observables localized on a light cone \cite{BdR}.  
In spite of the limitations just outlined, the strength of the DHR-analysis is that both for  BF-sectors  and BR-sectors  
the mathematical machinery of DHR-analysis associating the group of \emph{global} gauge transformations  to sectors applies\footnote{
It is worth mentioning that more works in AQFT dealt with general or peculiar facets of QED: e.g.\ the analysis of the infrared problem and the localization of electromagnetic charge\cite{B82, {BDMRS}}, the interpretation of QED by local constrains \cite{GLl}, lattice approach \cite{KRT}, definition and superselection of free models \cite{BDLR,C09, DL}, or proposal for interacting models by warped convolution \cite{BLS}.
}.
\smallskip 

\indent What remains to be understood is the \emph{local} gauge principle.\smallskip 

The abstract geometric formulation of the DHR analysis given by Roberts in terms of non-Abelian cohomology \cite{RobI, RobII} provide a steering for this quest:
it is motivated by the observation that the charge transporters of the global gauge theory satisfy a 1-cocycle equation. 
According to this,  Roberts also proposed that a  2-cohomology underlies the local gauge principle in  AQFT \cite{Rob77}: 
starting from a unitary 2-cocycle $\rw$ associated with the electromagnetic field,  one should recover a  potential $\ru$, i.e.\ the primitive   of $\rw$: local gauge transformations should arise as equivalence transformations of  $\ru$. Roberts also suggested that the charge transporters of the DHR-analysis should be 
replaced in the QED-case by a ``formal" gauge invariant expression as   
\begin{equation}
\label{uno}
\psi(\partial_0 p) \ e^{i \int_p A_\mu \, dp^\mu} \ \psi(\partial_1 p)^* \ ,
\end{equation}
where $p : \partial_1p \to \partial_0p$ is a path, $A_\mu$ is the electromagnetic potential and $\psi$ is the Dirac field, a sort of 
a finite Mandelstam string \cite{M62}.

\indent Along the ideas of non-Abelian cohomology, a further step has been done in \cite{RR, RRV}  where a geometrical interpretation of the above picture
has been furnished in terms of connections:  the potential $\ru$ can be interpreted as a connection having $\rw$ as curvature, and the transformations of the principal bundle associated with $\ru$ as local gauge transformations. But, from a physical point of view, 
the r\^ole played in gauge theories by observables associated to closed paths (Wilson loops)  suggested a finer comprehension  in terms of a net of local algebras\footnote{
Among the others, we recall two interpretations of observables localized on lines and loops (holonomies), broadly related to the present work.
Ashtekar and Corichi \cite{ACb, ACa} deepened the significance of the topological invariant of the Gau\ss\ linking number and its relation with the Fock-space inner product of the represented fields, \cite{S06} is a recent geometrical survey.
Moreover, the interpretation of electromagnetic charges in terms of the Wilson-'t Hooft operators in TQFT, e.g.\ the Kapustin-Witten lectures on electro-magnetic duality \cite{KW}.
}.

\indent  Actually, this is the motivation of our paper \cite{CRV}.
The \emph{net of causal loops} results to be a combinatorial, model independent construction of a covariant net of local $\rC^*$-algebras 
over any 4-d globally hyperbolic space-time, i.e.\ covariant with respect to the global symmetries of the space-time and respecting  the causality principle.
The generators of these  local algebras are closed paths (loops) associated with a suitable base $K$ of the underlying  space-time. 
Considering $K$ as a partially ordered set with respect to the inclusion, one can define a simplicial set.  
Paths turn out to be compositions of 1-simplices of this simplicial set,
and can be figured out as finite coverings, by elements of $K$, of curves of the space-time. \\
\indent  The very relation with local gauge theories, along the ideas drown by Roberts but also with significative  differences,  arose from the representations of this net. 
Covariant representations  turn out to be causal and covariant  2-cochains $\rw$, i.e.\ unitary valued functions of 2-simplices  localized on the boundary, not necessarily  satisfying a cocycle equation. Furthermore, any such 2-cochain $\rw$ induces, by a procedure akin  to the reconstruction of the primitive 
of  an exact  2-form, a causal and covariant connection system $\ru$: a family of connections, one for each element of $K$, in which causality and covariance arise as properties of the system and not of a single connection. Local gauge transformations naturally arise as maps between equivalent connection systems.\smallskip

The present paper completes the abstract arguments in \cite{CRV}, realizing QED as representation of the net of causal loops in Minkowski 
space-time. More precisely,  we shall map the quantum electromagnetic field $F_{\mu\nu}$, \emph{not free in general}, into 
a representation of the net of causal loops and show that all the abstract notions just described  find a counterpart in terms of $F_{\mu\nu}$.  \smallskip

 The map is constructed using a 2-form $F_{\mu\nu}(y)$ associated with the electromagnetic field $F_{\mu\nu}$ through the convolution of the field 
 with a test function, Section \ref{0B}. Using a contracting homotopy of the Minkowski space-time, we reconstruct the primitive 1-form $A^{z}_{\mu}(y)$, called  the electromagnetic potential form with respect to $z$, the point which the space-time is contracted to. 
The potential form $A^{z}_{\mu}(y)$, closely related to the quantum electromagnetic potential field, is neither local nor covariant.  
However, causality and covariance arise without recurring to the Gupta-Bleuler formalism, as properties of the system $A_\mu:=\{A^{z}_{\mu} \ , \ z\in \bR^4\}$ of potential forms\footnote{A covariant quantum electromagnetic potential was find in \cite{MSY} for the case of infinite string-localized fields.}.  
The notion of local gauge transformation is recast in terms of this system. \\
\indent We then consider, in Section \ref{0C},  the net of causal loops defined over the set of double cones $K$ of the Minkowski space-time. We slightly modify its definition with respect to \cite{CRV} since we take into account two relevant facts: the flatness of the Minkowski space-time and the necessity of  smearing the fields.
The essential effect of this choice is to adjust the definition of the simplicial set, replacing the r\^ole played by  the elements of $K$  by the set of test functions supported within elements of $K$.\\ 
\indent Practically the construction runs as follows: for any 2-simplex  we associate the exponential of the integral of $F_{\mu\nu}(y)$ over the triangular surface underling the 2-simplex. Thanks to the Stokes' theorem this defines a function $\rw^{em}$ over the
2-simplices which has the right localization properties, turning to define a representation of the net of causal loops.
Yet, we can define two connection systems: $\ru^{em}$ which is  obtained from $\rw^{em}$ by using the the combinatorial procedure outlined above; and  $\ru^{pot}$ obtained as the exponential of the line integral of $A_\mu$. 
The key result is that $\ru^{pot}$ is equivalent to $\ru^{em}$ in terms of a local gauge transformation. 
Moreover, we show that local gauge transformations of $A_\mu$ agree with local gauge transformations of the connection systems.\smallskip

Comparing our results with the Roberts' approach recalled above,  the unitary 2-cochain $\rw^{em}$ associated with the 
electromagnetic field is \emph{not}, in general, a 2-cocycle. The relevant physical information is carried by the localization property of $\rw^{em}$,  as in the very philosophy of AQFT. This is enough to recover not a single connection but a \emph{system} of connections in which causality and covariance however arise.
Moreover, it suggests to replace the electromagnetic \emph{potential field} by the electromagnetic \emph{potential form} in the expression (\ref{uno}) for the charge transporters of QED, so  we have 
\begin{equation}
\label{due}
\psi(\partial_0p) \ e^{i\int_p A^{z}_\mu \, dp^\mu} \ \psi(\partial_1p)^* \ .
\end{equation}
Hence we have a direct dependence on the point $z$ of the space-time labelling the potential form.
This parallels what happens in the analysis of sectors associated to charges of electromagnetic type \cite{BdR}, 
in which the theory is developed on the forward light-cone of a single point of the space-time, as outlined above.

\section{Preliminaries} 
\label{0A}
In this section we describe the setting and the notation used in the present paper.  In the first subsection  we recall some facts 
concerning the Minkowski space-time, the Poincar\'e group   and  the notion of causal nets of $\rC^*$-algebras. Some properties 
about quantum electromagnetic field, relevant for the aims of the present paper,  are discussed in the second subsection. Finally,  
we shall deal with the integration of forms in the language  of singular homology.

\subsection{Nets of $\rC^*$-algebras on the Minkowski space-time}

\paragraph{The Minkoswki space-time  and the Poincar\'e group.} We recall some basic properties of the Minkowski space-time $(\bR^4,g)$  and establish some notations.  
We adopt the convention that the metric tensor $g_{\mu\nu}$ has negative signature:  $g_{00}=1$, $g_{11}=g_{22}=g_{33}=-1$ and  $g_{\mu\nu}=0$ if  $\mu\ne \nu$. 
Furthermore, recall that $g_{\mu\nu}=g^{\mu\nu}$ and ${g^\mu}_\nu={g_\mu}^\nu=\delta_{\mu\nu}$.  
We use the contravariant notation to represent the components of element $x$ of the Minkowski space-time:  $x=(x^0,x^1,x^2,x^3)$ or  $x=x^\mu$, with $\mu=0,1,2,3$.  The inner product induced by the metric tensor   is 
\[
x\cdot y:= x^0y^0-x^1y^1- x^2y^2-x^3y^3= x^\mu g_{\mu\nu} x^\nu = x^\mu x_\mu \ ,
\]
where $x_\mu= g_{\mu\nu} x^\nu$ is the covariant representation of $x$. Clearly $x^\mu= g^{\mu\nu} x_\nu$. Denoting  the canonical scalar product of two element $x,y$ of $\bR^4$ by 
$(x,y)$ and by $g$ the matrix associated with the metric tensor $g_{\mu\nu}$,   
the Minkowski inner product can be rewritten as $x\cdot y = (x, gy)$.\\
\indent We shall say that two subsets $X$ and $Y$ of $\bR^4$ are \emph{causally disjoints} if, and only if,  the corresponding elements are spacelike separated. In symbol we write  
\[
 X \perp Y \ \ \iff  \ \ (x-y)^2 <0 \ , \qquad  \forall x\in Y, \ y\in Y \ . 
\]
A Lorentz transformation is a linear transformation $L$ leaving the inner product invariant,
 $(Lx, g Ly) = (x,g y)$:  
in matrix notation  $L^TgL=g$. 
In tensor notation, if $(Lx)^\mu= {L^\mu}_\nu x^\nu$ and $(Lx)_\mu= {L_\mu}^\nu x_\nu$,  where ${L_\mu }^\nu = g_{\mu\alpha} {L^\alpha}_\beta g^{\beta\alpha}$, then 
${L^\alpha}_\mu g_{\alpha\beta} {L^\beta}_\nu  = g_{\mu\nu}$. So
\[
 {(L^{-1})^\nu}_\alpha = g_{\alpha\beta} {L^\beta}_{\rho} g^{\rho\nu} = {L_\alpha}^{\nu} .
 \]  
The \emph{restricted  Lorentz group}  $\sL^\uparrow_+$  is  the subgroup of the Lorentz  transformations whose matrices ${L^\mu}_\nu$
have positive determinant and ${L^0}_0\geq 1$. The \emph{Poincar\'e group} $\sP^\uparrow_+$ is the semi-direct product $\bR^4\rtimes  \sL^\uparrow_+$ with composition law  defined as 
\[
(x,L)\, (x',L'):= (x+Lx', LL') \ . 
\]
According to this relation, $(-L^{-1}x, L^{-1})$ is the inverse of $(x,L)$. To economize notation, sometimes we shall denote an element of the Poincar\'e group 
by $P:=(x,L)$, so when $Y \subset \bR^4$ or $y \in \bR^4$ we shall write the action as 
\[
PY = x + LY \ \ , \ \ Py = x + Ly \ .
\]
\paragraph{Nets of $\rC^*$-algebras.} The mathematical description of local observables in AQFT is given 
in terms of nets of $\rC^*$-algebras.  We shall focus on the case of the Minkowski space-time and refer the reader to \cite{RV, CRV} for more general situations. 
Let $K$ denote the set of \emph{double cones} of the Minkowski space-time \cite{Haa}. Double cones form a base of the topology of  $\bR^4$ which is stable under the action of the Poincar\'e group and   upward directed under inclusion. 
By a \emph{net} of $\rC^*$-algebras over the Minkowski space-time we shall mean an inclusion preserving (\emph{isotonous})  correspondence 
$\cA:K\ni o\to \cA_o\subseteq \cA(\bR^4)$, i.e.\ 
\[
o_1\subseteq o_2 \ \ \Rightarrow \ \ \cA_{o_1}\subseteq \cA_{o_2} \ ,  
\]
associating a $\rC^*$-subalgebra $\cA_o$ of  a fixed target $\rC^*$-algebra $\cA(\bR^4)$ to any double cone $o$.
The net $\cA$ is said to be \emph{causal} whenever 
\[
o_1\perp o_2 \ \ \Rightarrow \ \ [\cA_{o_1},\cA_{o_2}]=0 \ , 
\]  
and it is said to be \emph{covariant} if there is action of the Poincar\'e group $\alpha:\sP^\uparrow_+\to\mathbf{aut}\cA(\bR^4)$ such that 
\[
\alpha_P\circ \cA =\cA\circ \alpha_P \ ,\qquad P\in \sP^\uparrow_+ \ . 
\]
In the following we shall denote a causal and covariant net of $\rC^*$-algebras over the Minkowski space-time by $(\cA,\alpha)_K$.

\subsection{The quantum electromagnetic field.}
\label{0Aa}
The quantum electromagnetic field is defined by the following data. 
\begin{itemize}
\item{}  A Hilbert space $\cH$ carrying a unitary representation $U$  of $\sP^\uparrow_+$ having a unique 
invariant (\emph{vacuum}) vector $\Omega$ and satisfying the spectrum condition.
\item{} An operator  valued distribution $f\mapsto F_{\mu\nu}(f)$ assigning to any real test function $f\in\cS(\bR^4)$ an essentially self-adjoint operator $F_{\mu\nu}(f)$ of $\cH$ defined on a 
dense \emph{domain} $\cD\subset \cH$ containing  the vacuum and  such that    $ F_{\mu\nu}(f)\cD\subseteq \cD$ for any test function $f$. We furthermore assume that the essential self-adjointness property is conserved under the contraction over 1- and 2-tensor taking value in $\cS(\bR^4)$.   
\item{} \emph{Covariance.} The representation $U$ leaves the domain $\cD$ invariant and 
\begin{equation}
\label{0A:1}
 U(x,L)  F_{\mu\nu}(f) U(x,L)^*  \ = \ {{L^{-1}}_\mu}^\alpha  {{L^{-1}}_\nu}^\beta \, F_{\alpha\beta}(f_{(x,L)}) 
 \ , \qquad f\in\cS(\bR^4) \ , 
\end{equation}
for any $(x,L)\in\sP^\uparrow_+$, where $f_{(x,L)}:= f\circ (x,L)^{-1}$. 
\item \emph{(strong) Causality.}   Causality of the field reads as  
\[
  \left[ F_{\mu\nu}(f),  F_{\mu\nu}(g)\right]=0 \ \ , \ \qquad \supp(f)\perp \supp(g) \ , 
\] 
where the commutation relation is intended to hold on $\cD$, and the symbol $\perp$ stands for causal disjointness. However in the present paper we need a stronger relation,
\begin{equation}
\label{0A:2}
\left[ \exp(i F_{\mu\nu}(f)),\exp(i F_{\alpha\beta}(g))\right]=0   \ \ , \ \qquad \supp(f)\perp \supp(g) \ , 
\end{equation}
where, to ease notation,  we are denoting $F_{\mu\nu}(f)$ and its closure by the same symbol. 
\item (\emph{Field equations}) The field is \emph{antisymmetric}  
\begin{equation}
\label{0A:3}
F_{\mu\nu} = - F_{\nu\mu} \ , 
\end{equation}  
and satisfies the  \emph{Maxwell equations}: the \emph{first} 
\begin{equation}
\label{0A:4}
\partial_{\si} F_{\mu\nu} + \partial_\mu  F_{\nu\si} + \partial_\nu  F_{\si\mu} = 0  \ , 
\end{equation}
and the \emph{second} 
\begin{equation}
\label{0A:5}
\partial^\mu F_{\mu\nu} = J_\nu \ .  
\end{equation}
\end{itemize}
Let us briefly comment the terminology used in the above definitions. As usual, the term \emph{operator valued distribution} means that the mapping $f\mapsto (\phi,F_{\mu\nu}(f)\psi)$ is a tempered distribution, i.e.\ a continuous linear functional on the Frechet space $\cS(\bR^4)$ for any $\phi,\psi\in\cD$. The first and the second Maxwell equation hold in the weak sense i.e.\  $\partial_\alpha F_{\mu\nu}(f) := - F_{\mu\nu}(\partial_\alpha f)$. The explicit construction of $(\cH,F,U)$ can be made, in the free case, by using the Fock spaces \cite{Fre10} or the Wightman reconstruction theorem \cite{Ste}.\\
\indent We discuss the assumptions about the second Maxwell equation and the strong form of causality. As observed we are not assuming that 
the electromagnetic field is free, that is we allow $J_\nu \ne 0$. This is possible since the second Maxwell equation does not enter directly in the definitions and constructions of the present paper. 
Concerning the strong form of causality, this is necessary since using the exponential of the fields we will constructs representations of the net of \emph{causal} loops. This strong form of causality is verified by the free electromagnetic field \cite[Section 5.3]{Ste}  and, also, by some interacting fields \cite{GJ}. As an example  
one may consider the free electromagnetic field $F^0_{\mu\nu}$ (so $\partial^\mu F^0_{\mu\nu}=0$) coupled with an external classical current $j_\nu$ \cite{Fre10}:  
one  considers  a bounded closed covariant 2-form $\phi_{\mu\nu}$, i.e.\  fulfilling   
\[
\phi_{\mu\nu}\circ P^{-1}= {L^{-1}_{\mu}}^\alpha {L^{-1}_{\nu}}^\beta \phi_{\alpha\beta} \ , \qquad P\in \sP^\uparrow_+ \ . 
\]
Then,  the field $F_{\mu\nu}$, defined by 
\[
 F_{\mu\nu}(f):=  F^0_{\mu\nu}(f)+ \int \phi_{\mu\nu}(x)f(x)\, d^4x \ , \qquad f\in\cS(\bR^4) \ , 
\] 
satisfies the above properties with  current $j_\nu:=\partial^\mu\phi_{\mu\nu}$.

\subsection{Singular simplices and integration}
\label{0Ab}
The suited language for dealing with the Stokes' theorem is that of the algebraic topology. We start by introducing the simplicial set 
of singular piecewise smooth simplices in $\bR^4$. Given the  \emph{standard $n$-simplex}
\[
\Delta_n \ := \ \left\{ (t_1,t_2,\ldots,t_n)\in\bR^n \ \big| \ \ t_i\geq 0,\ \sum^n_{i=1} t_i\leq1\right\} \ , 
\] 
a \emph{singular piecewise smooth $n$-simplex} $\chi$ is a piecewise smooth maps $\chi:\Delta_n\to \bR^4$.  The vertices of $\chi$ are the images of the vertices of $\Delta_n$: so, given the canonical base $\{ e_k \}$ of $\bR^n$ and $i \in \{ 0,\ldots,n \}$, the $i$-\emph{vertex} of $\chi$ is defined by $\chi(e_i)$ when $i \in \{ 1,\ldots,n \}$, and by $\chi(0)$ when $i=0$.
The order of the vertices endows $\chi$ with a natural orientation.
We say that an $n$-simplex $\chi'$ has \emph{opposite orientation} with respect to $\chi$ whenever $\chi'=\chi\circ T$, where $T:\bR^n\to\bR^n$ is an affine transformation making an odd permutation of the vertices of $\Delta_n$. 
Denoting the set of singular piecewise smooth $n$-simplices by $\Si_n(\bR^4)$, the \emph{face maps}   $\partial_i:\Si_n(\bR^4)\to \Si_{n-1}(\bR^4)$  are defined by  
\begin{equation}
\label{0A:5a}
\partial_i\chi (t_1,\ldots, t_{n-1}):= \left\{
\begin{array}{ll}
\chi (1- \sum^{n-1}_{i=1}t_i,\, t_1,t_2,\ldots,t_{n-1}) \ ,  & i=0  \ , \\
\chi (t_1,\cdots, t_{i-1},0,t_i,\dots,t_{n-1}) \ , & i >0 \ . 
\end{array}
\right.
\end{equation}
One can easily see that the collection $\Si_*(\bR^4) = \{ \Si_n(\bR^4) , n\in\bN \}$, is a simplicial set.

\paragraph{Singular Homology.} We denote the $\bZ$-module of \emph{singular $n^{th}$-chains}  by $C_n(\bR^4)$: this is the set of finite formal linear combinations $\sum_i m_i \chi_i$ with $m_i\in\bZ$ and $\chi_i\in \Si_n(\bR^4)$. The \emph{boundary} of $\chi \in \Si_{n+1}(\bR^4)$ is the $n$-chain
\begin{equation}
\label{0A:6}
\partial \chi \ = \ \sum^{n+1}_{i=0} (-1)^i \, \partial_i \chi  \ ,
\end{equation}
and an $n$-\emph{cycle} is an $n$-simplex $\chi$ such that $\partial\chi=0$. 
Now, as well known, since  $\partial\circ\partial=0$ it turns out that \emph{the set $B_n(\bR^4)$ of boundaries} of  $C_{n+1}(\bR^4)$ is a submodule of \emph{the set $Z_n(\bR^4)$ of $n$-cycles}, and $H_n(\bR^4):=Z_n(\bR^4)/B_n(\bR^4)$ is the $n$-th module of singular homology of $\bR^4$, which vanishes for $n>0$. The proof relies on the existence of  contracting homotopies  for  $\bR^4$. The easiest example is the 
\emph{cone construction} which  associates to any $z \in \bR^4$ a map 
$h^z : C_n(\bR^4)\to C_{n+1}(\bR^4)$
such that
\begin{equation}
\label{eq.h}
\partial\circ h^z +h^z\circ \partial \ = \ id_{C_n(\bR^4)} \ .
\end{equation}
Explicitly, $h^z$ is defined on $\chi \in \Si_n(\bR^4) \subset C_n(\bR^4)$ by
\begin{equation}
\label{0A:8a}
 h^z\chi(t_1,\ldots,t_n, t_{n+1}) \ := \ 
 \left(\sum^{n+1}_{i=1} t_i\right)\chi\left( \frac{t_2}{\sum^{n+1}_{i=1} t_i}, \dots, \frac{t_{n+1}}{\sum^{n+1}_{i=1} t_i}\right) +\left(1-\sum^{n+1}_i t_i\right) z 
 \ , 
\end{equation}
and extended by linearity on $n$-chains. In particular $\partial_0 h^z\chi= \chi$ and $\partial_i h^z\chi= h^z\partial_{i-1}\chi$ for $i>0$.  
This implies that any $\chi \in C_1(\bR^4)$ which is a closed curve ($\chi(0)=\chi(1)$) is the boundary 
of the  2-simplex $h^z\chi$.

\paragraph{Integration.} Given a smooth $n$-form $\omega= \sum_{\alpha_1\cdots \alpha_n} \omega_{\alpha_1\cdots \alpha_n} dx^{\alpha_1}\wedge\cdots \wedge dx^{\alpha_n}$,  
where the indices $\alpha_1,\ldots,\alpha_n$ vary independently in $0,1,2,3$,  the integral over an $n$-simplex is defined by  
\begin{equation}
\label{0A:7}
 \int_\chi \omega \ := \ \int_{\Delta_n} \omega_{\alpha_1\cdots\alpha_n}(\chi(t))\, \chi^{\alpha_1\cdots\alpha_n} (t) \,  d^n t \ , 
 \qquad \chi\in\Si_n(\bR^4) \ , 
\end{equation}
where $\chi^{\alpha_1\cdots\alpha_n}(t)$ is the involved Jacobian.
The Stokes' theorem hence reads as follows: given an $(n-1)$-form $\omega$ and $\chi \in \Si_n(\bR^4)$, then
\begin{equation}
\label{0A:8}
 \int_\chi \rd \omega \ = \ 
 \int_{\partial \chi}  \omega \ = \
 \sum^n_{i=0} (-1)^n \int_{\partial_i\chi} \omega \ , \qquad 
 \chi\in \Si_n(\bR^4) \ . 
\end{equation}
If $\varphi = \sum_j m_j \chi_j \in C_n(\bR^4)$ and $\omega$ is an $n$-form, then 
$\int_\varphi \omega  := \sum_j m_j \int_{\chi_j} \omega$ and the Stokes' theorem extends to chains by linearity. \\
\indent In what follows  we shall deal mainly with singular $0$-, $1$-, $2$- simplices in $\bR^4$:  a $0$-simplex is a point $x \in \bR^4$; a 1-simplex is a piecewise smooth curve $\gamma$;  a 2-simplex is a piecewise smooth surface $\si$. Note that the boundary of a curve $\gamma$ is the $0$-chain $\partial\gamma=\partial_0\gamma -\partial_1\gamma$ that in terms of vertices is $\gamma(1)-\gamma(0)$.  The boundary of  a surface $\si$ is the 1-chain $\partial\si= \partial_0\si-\partial_1\si+\partial_2\si$. 
According to the above definition of orientation, we note that there is only one curve  having opposite orientation of a given curve $\gamma$; this is the  curve 
$\bar{\gamma}(t)=\gamma(1-t)$, $t\in\Delta_1$, that we call the \emph{opposite} of $\gamma$. Note that $\partial\bar\gamma=-\partial\gamma$. 
Instead in the case of surfaces the opposite is not unique: as a convention, we call the \emph{opposite} of $\si$ the surface $\bar\si:= \si\circ T$, where  
$T:=\left(\begin{smallmatrix}
0&1\\ 1&0
\end{smallmatrix} \right)$
is written in the canonical base of $\bR^2$. The boundary of $\bar\si$ results to be the 1-chain 
\begin{equation}
\label{0A:9}
\partial\bar\si \ = \
\partial_0\bar\si-\partial_1\bar\si+\partial_2\bar\si \ = \
-\partial_0\si -\partial_2\si+ \partial_1\si \ = \
-\partial \sigma \ . 
\end{equation}
Concerning the integration over singular simplices, we observe that the infinitesimal line element of $\gamma \in \Si_1(\bR^4)$ is $\dot\gamma^\mu(t) dt$; the infinitesimal surface element of $\si \in \Si_2(\bR^4)$ is $\si^{\mu\nu}(t) d^2t$, where
\begin{equation}
\label{0A:10}
\si^{\mu\nu}(t) \ = \ 
\left(\frac{\partial\si^\mu}{\partial t_1} \frac{ \partial \si^\nu }{ \partial t_2 } -
\frac{ \partial \si^\mu }{ \partial t_2 }\frac{ \partial \si^\nu }{ \partial t_1 }\right)(t)
\ \ , \ \  t \in \Delta_2 \ .
\end{equation}
Note that $\dot\gamma^\mu$ and $\partial \si^\mu / \partial t_k$ may have discontinuity points, nevertheless any of them is bounded.

\section{Loops, surfaces and the potential system} 
\label{0B}

In this section we define the integral of the electromagnetic field over a 2-simplex,
and prove that this operator is covariant in a suitable sense and localized on the boundary. 
To this end, we shall use the fact that the electromagnetic field defines a closed, exact, 2-form and the Stokes' theorem. 
Furthermore we shall consider a collection of primitives of the above 2-form, indexed by the Minkowski space-time, and analyze its properties.

For what follows we shall use the language of algebraic topology, in which the integrations of smeared fields over surfaces and over curves are treated on the same footing and where the Stokes' theorem applies in a more general form. To this end, we shall use smearing functions to generalize the notion of $n$-simplex.
To keep contact with the geometrical intuition we shall mainly work with $n$-simplices rather than $n$-chains, but our results easily extend to these latter.

\subsection{Smearing simplices} 
\label{0AB}

Both for the purpose of smearing the fields and for treating abstract simplices given by subsets of the space-time, as in \cite{CRV}, 
we refine the notion of piecewise smooth $n$-simplex. 
We define the set of (singular, piecewise smooth) \emph{smearing $n$-simplices} by 
\begin{equation}
\label{0A:11}
\Si_n(\bR^4,\cS) \ := \ \{ (\chi,f) : \chi \in \Si_n(\bR^4) \ , \ f \in \cS(\bR^4) \, , \, \supp(f) \ni 0  \} \ , 
\end{equation}
and call $f$ the \emph{smearing function} of $(\chi,f)$. The  Poincar\'e group acts on $\Si_n(\bR^4,\cS)$ by   
\begin{equation}
\label{0A:12}
P(\chi,f) \ := \ ( P\chi , f_{L} ) \ \ , \qquad 
P=(x,L) \in \sP^\uparrow_+  \ , \ (\chi,f) \in \Si_n(\bR^4,\cS) \ ,
\end{equation}
where $(P\chi(t))^\mu := x^\mu+ {L^\mu}_\nu \chi^\nu(t)$, with $t\in\Delta_n$, and $f_{L}\equiv f_{(0,L)}$. Notice that translations do not act on the smearing  function $f$, so $\supp(f_{L}) \ni 0$.\\ 
\indent The $\bZ$-module of \emph{smearing $n$-chain} $C_n(\bR^4,\cS)$ is defined as the set of  finite formal linear combinations of smearing $n$-simplices,
\begin{equation}
\label{0A:12a}
C_n(\bR^4,\cS) \ := \ 
\left\{ \sum_i m_i (\chi_i,f_i) \ | \  m_i\in\bZ \ , \  (\chi_i,f_i)\in\Si_{n}(\bR^4,S) \right\} \ ,
\end{equation}
and the Poincar\'e action extends on  $C_n(\bR^4,\cS)$ correspondingly. \\
\indent The face maps lift, naturally,  from 
$\Si_n(\bR^4)$ to $\Si_n(\bR^4,\cS)$  by defining $\partial_i (\chi,f):=(\partial_i\chi, f)$ for any  $i$. So the \emph{boundary} of a smearing $n$-simplex is the $(n-1)$-chain  
\[
 \partial(\chi,f) \ := \
 %
 %
 (\partial \chi,f) \ . 
\]
Similarly to Subsection \ref{0Ab} we define   $H_n(\bR^4,\cS):=Z_n(\bR^4,\cS)/B_n(\bR^4,\cS)$ that also vanish for $n>0$. 
Notice that also in this case $n$-boundaries coincide with $n$-cycles and that the  Poincar\'e action leaves them invariant.\\ 
\indent The reason why we introduce smearing simplices  relies on the following observation: to any smearing $n$-simplex $(\chi,f)$ there corresponds an $n$-form $\chi[f]$ defined by  
\begin{equation}
\label{0A:13}
\chi[f]^{\alpha_1\cdots\alpha_n}(x) \ := \ \int_{\Delta_n} f(x-\chi(t))\, \chi^{\alpha_1\cdots\alpha_n}(t)\, d^nt 
\ \ , \ \
x \in \bR^4
\ . 
\end{equation}
This $n$-form plays the r\^ole of  the $n$-volume element for the integration of a quantum field over an $n$-simplex.   
It is easily seen that $\chi[f]^{\alpha_1\cdots\alpha_n}\in\cS(\bR^4)$  and that  the partial Riemann sums defining the integral (\ref{0A:13}) 
converge in the sense of the topology of $\cS(\bR^4)$ \cite{Rudin}. It is natural to define the \emph{support} of a smearing $n$-simplex as the support of the corresponding $n$-form, i.e.\  
\begin{equation}
\label{0A:14}
|(\chi,f)| \ := \ \supp(\chi[f])\ , \qquad (\chi,f)\in \Si_n(\bR^4,\cS) \ . 
\end{equation}
It is easy to verify that the support is Poincar\'e covariant, $|P(\chi,f)|= P|(\chi,f)|$ for $P\in\sP^\uparrow_+$,  and that 
if $\supp(f)$ is contained in the open ball $B_\eps(0) \subset \bR^4$, $\eps > 0$, then 
\begin{equation}
\label{0A:15}
\supp (\chi[f]) \ \subseteq \  B_\eps(\chi) \ := \ \{ x \in \bR^4 : d(x,\chi(\Delta_n)) < \eps \} \ , 
\end{equation}
in fact $x \notin B_\eps(\chi)$ implies $f( x-\chi(t) ) = 0$ for all $t \in \Delta_n$. Finally we observe that the above reasoning extends to $n$-chains, as follows: 
given an $n$-chain $\varphi=\sum_i m_i(\chi_i, f_i)$, we set
\begin{equation}
\label{0A:16}
[\varphi]^{\alpha_1\cdots\alpha_n}(x) \ := \
\sum_i m_i \, \chi_i [f_i]^{\alpha_1\cdots\alpha_n}(x)
\ \ , \ \
x \in \bR^4 \ ,
\end{equation}
which is still an $n$-form with coefficients in $\cS(\bR^4)$. 
The support of the $n$-chain is defined as the support of the associated $n$-form, and the property of covariance easily generalizes.

\subsection{The electromagnetic field 2-form}
\label{0Ba}
As already said, we want to integrate the electromagnetic field on a 2-simplex. 
The idea is to integrate on the given 2-simplex the convolution of $F_{\mu\nu}$ by a test function. To this end, for any $f\in\cS(\bR^4)$  define 
\[
f_y(x):=  f(x-y) \ ,  \qquad x,y\in\bR^4,
\]
and observe that $f_y$ is a test function as well for any $y$.  Since $y\mapsto f_y$ fits the topology of Schwartz functions, 
the mapping $y\mapsto F_{\mu\nu}(f_y)$  is smooth in the sense of distributions, i.e.\  the functions 
$y\mapsto (\phi,F_{\mu\nu}(f_y)\psi)$ are smooth for any $\phi,\psi\in\cD$ and for any $\mu,\nu$ (see \cite{Rudin}).  
Furthermore, covariance implies
\begin{equation}
\label{0B:1}
F_{\mu\nu}(f_y) \ = \ U(y) F_{\mu\nu}(f) U(y)^* \ , \qquad y\in\bR^4  \ , 
\end{equation}
where, with an abuse of notation, $U(y)$ stands for $U(y,\mathbbm{1})$. 
Finally, we note that $\partial_{y^\rho} f_y = - (\partial_\rho f)_y$ and this implies
$\partial_{y^\rho} F_{\mu\nu}(f_y) = F_{\mu\nu}((\partial_\rho f)_y)$.
So
\begin{align*}
\partial_{y^\rho} F_{\mu\nu}(f_y) + \partial_{y^\nu} F_{\rho\mu}(f_y)  + \partial_{y^\mu} F_{\nu\rho}(f_y) & =
   F_{\mu\nu}( (\partial_{\rho} f)_y) +  F_{\rho\mu}((\partial_\nu  f)_y)   + F_{\nu\rho}((\partial_\mu  f)_y) \\
& =  U(y) \big( F_{\mu\nu}(\partial_{\rho} f) +  F_{\rho\mu}(\partial_\nu  f)   + F_{\nu\rho}(\partial_\mu  f) \big) U(y)^* \\
& = 0 \ ,
\end{align*}
and this leads to 
\begin{equation}
\label{0B:2}
\partial_{y^\rho} F_{\mu\nu}(f_y) + \partial_{y^\nu} F_{\rho\mu}(f_y)   + \partial_{y^\mu} F_{\nu\rho}(f_y) =0 \ ,
\end{equation}
that is, $y\mapsto F_{\mu\nu}(f_y)$ is a closed 2-form in the sense of distributions.


\paragraph{The surface integral of $F_{\mu\nu}$.}  The rough idea is to define our integral on $\si\in\Si_2(\bR^4)$ as 
\begin{equation}
\label{0B:2o}
\int_\si F(f_\si) \, d\si \ := \ \frac{1}{2}\, \int_{\Delta_2} F_{\mu\nu}(f_{\si(t)})\,  \si^{\mu\nu}(t) \, d^2t \ ,
\end{equation}
where $f \in \cS$ and the factor $1/2$ derives from the fact that $F_{\mu\nu}$ is antisymmetric and $\mu,\nu$ varies independently in $\{ 0,1,2,3 \}$.
This is a formal expression since the integration of unbounded operators is involved. To deal with this problem we consider the smearing surface $(\si,f) \in \Si_2(\bR^4,\cS)$ and the corresponding 2-form
\begin{equation}
\label{0B:2c}
 \sigma[f]^{\mu \nu}(x)  \ := \  \int_{\Delta_2} f(x-\si(t)) \, \si^{\mu\nu}(t) \, d^2t \ , \qquad    x \in \bR^4 \ ,
\end{equation}
defined in accordance with (\ref{0A:13}).
We have seen that any $\sigma[f]^{\mu \nu}$ is in $\cS(\bR^4)$ and that the partial Riemann sums defining the integral (\ref{0B:2c}) converge in the Schwartz topology; furthermore, $\supp (\sigma[f]) \subseteq \  B_\eps(\sigma)$ whenever  $\supp(f)\subseteq B_\eps(0)$,  see (\ref{0A:15}).\smallskip 

On these grounds we define
\begin{equation}
\label{0B:4}
F\left<\si,f\right> \ := \ \frac{1}{2}\ F_{\mu\nu}(\sigma[f]^{\mu \nu}) \ \ , \qquad  (\si,f) \in \Si_2(\bR^4,\cS) \ .
\end{equation}
Next results show the properties of this operator. In particular, the following Lemma \ref{0B:5}$(ii)$ illustrates how (\ref{0B:4}) does fit 
the formal expression (\ref{0B:2o}).
%
\begin{lemma}
\label{0B:5} 
Let $(\si,f) \in \Si_2(\bR^4,\cS)$. Then   $F\left<\si,f\right>$ is an essentially self-adjoint operator defined on $\cD$ such that 
$F\left<\si,f\right>\cD\subseteq \cD$, and:
\begin{itemize}
\item[(i)]   $U(P) F\left<\si,f\right> U(P)^* \ = \ F\left<P(\si,f)\right>$,   \ \ $P\in \sP^\uparrow_+$;  
\item[(ii)]  $ (\phi,F\left<\si,f\right>\psi)  \ = \ 1/2 \, \int_{\Delta_2} (\phi,F_{\mu\nu}(f_{\si(t)})\psi) \, \si^{\mu \nu}(t) \ d^2t$,  \ \  $\phi,\psi\in\cD$; 
\item[(iii)] if  $|(\si',f')|\perp |(\si,f)|$, then  $\big[\exp ( iF\left<\si,f\right> ),\exp ( iF\left<\si',f'\right> ) \big]  = 0$,
             where $|\cdot |$ denotes the support (\ref{0A:14}) of the smearing surface.
\end{itemize}              
\end{lemma}
\begin{proof}
The first properties follow by definition. 
$(i)$ Let $P = (x,A) \in \sP^\uparrow_+$.
According to (\ref{0A:1}) we have $U(P) F\left<\si,f\right> U(P)^* \ = (1/2)\ {{L^{-1}}_\mu}^\alpha\, {{L^{-1}}_\nu}^\beta \, F_{\alpha\beta}(\sigma[f]^{\mu \nu}\circ P^{-1})$, 
where 
\begin{align*}
\sigma[f]^{\mu \nu}\circ P^{-1}(y) & = 
\int_{\Delta_2} f(L^{-1}y - L^{-1}x-\si(t) ) \, \si^{\mu \nu}(t) \ d^2t \\ & =
\int_{\Delta_2} f_L(y- x-L\si(t) ) \, \si^{\mu \nu}(t) \ d^2t =
\int_{\Delta_2} (f_L)_{x+L\si(t)}(y)  \, \si^{\mu \nu}(t) \ d^2t \ . 
\end{align*}
Since ${{L^{-1}}_\mu}^\alpha= {L^\alpha}_\mu$ (see Section \ref{0A}), we get 
\begin{align*}
2\cdot U(P) &  F\left<\si,f\right> U(P)^*  = \\ 
& = \int_{\Delta_2}  F_{\alpha\beta}((f_L)_{x+L\si(t)}) \cdot {{L^{-1}}_\mu}^\alpha\, {{L^{-1}}_\nu}^\beta\, \si^{\mu \nu}(t) \ d^2t 
 =  \int_{\Delta_2}  F_{\alpha\beta}((f_L)_{x+L\si(t)}) \cdot (L\si)^{\alpha\beta} (t) \ d^2t  \\
 & =  \int_{\Delta_2}  F_{\alpha\beta}((f_L)_{x+L\si(t)}) \cdot (x+L\si)^{\alpha\beta} (t) \ d^2t \\
& = F_{\alpha\beta}({P\sigma}[f_L]^{\alpha\beta}) \ = \ 2\cdot F\left<P(\si,f)\right> \ ,
\end{align*}
completing the proof. 
$(ii)$ follows from the observation that the map
$y \mapsto \int_{\Delta_2} f_{\si(t)}(y) \si^{\mu\nu}(t) d^2t$ 
is limit in the topology of $\cS(\bR^4)$ of finite Riemann sums, so applying the distribution 
$(\phi,F_{\mu\nu}(\cdot)\psi)$ 
we can pass under the sign of integral.
$(iii)$ follows by the definition (\ref{0B:4}) and strong causality (\ref{0A:2}).
\end{proof}

\paragraph{Cycle phases and the Stokes' theorem.}  We now construct a causal and covariant map from smearing 1-boundaries into the unitary group of the vacuum Hilbert space.
%
\begin{lemma}
\label{lem.0B:6} 
If $(\si,f),(\si',f) \in \Si_2(\bR^4,\cS)$ have the same  boundary then 
$F\left<\si,f\right>  =  F\left<\si',f\right>$.
\end{lemma}
\begin{proof}
Equation (\ref{0B:2}) implies that for any $\phi,\psi\in\cD$ the 2-form 
\[
\omega_{\mu\nu}(y) \ := \ (\phi,F_{\mu\nu}(f_y)\psi) \, , \quad y \in \bR^4 \ ,
\]
is closed, so it is exact and we can find a 1-form $\tau$ such that $d\tau=\omega$. 
Applying Lemma \ref{0B:5}.$(iii)$ and the Stokes' theorem (\ref{0A:8}) we get 
\begin{align*} 
(\varphi,F\left<\si,f\right>\psi)& = \frac{1}{2} 
\int_{\Delta_2} (\varphi,F_{\mu\nu}(f_{\si(t)})\psi) \, \si^{\mu \nu}(t) \, d^2t =  \frac{1}{2}\int_{\Delta_2} \omega_{\mu\nu} (\si(t))\, \si^{\mu \nu}(t) \, d^2t  
 =  
\int_\si \omega = \int_{\partial\si} \tau \\
&  = \int_{\partial\si'} \tau = \ \int_{\si'} \omega  =  
\frac{1}{2} \int_{\Delta_2} \omega_{\mu\nu} (\si'(t))\, \si'^{\mu \nu}(t) \ d^2t   =  
\frac{1}{2}\int_{\Delta_2} (\varphi,F_{\mu\nu}(f_{\si'(t)}) \psi) \si'_{\mu \nu}(t) \ d^2t \\
& = 
(\varphi,F\left<\si',f\right>\psi) \ , 
\end{align*}
and the proof follows by density of $\cD$ in $\cH$. 
\end{proof}
After a preliminary remark we illustrate a consequence of Lemma \ref{lem.0B:6}. By the cone construction (\ref{0A:8a}),
for any smearing closed curve $(\gamma,f)$, $\gamma(0)=\gamma(1)$, there is a 
smearing surface $h^z(\gamma,f):= (h^z\gamma,f)$ such that  $\partial h^z\gamma=\gamma$ for some $z\in\bR^4$.  
This suggests the following notation: 
we denote the set of \emph{closed smearing curves} by  $\partial\Si_2(\bR^4,\cS)$: clearly, this is a subset of $Z_1(\bR^4,\cS)$. 
\begin{theorem}[Cycle phases]
\label{thm.0B:6}
The quantum electromagnetic field induces the map
\[
w : \, \, \partial\Si_2(\bR^4,\cS) \to \cU\cH
\ \ , \ \ 
w(\ell,f) \, := \, \exp(iF\left<\si,f\right>)
\ , \
\ell = \partial \si
\ ,
\]
which is causal and covariant in the sense that
\[
[ \, w(\ell_1,f_1) \, , \, w(\ell_2,f_2) \, ] \ = \ 0 
\ \ \ , \ \ \
U(P) \, w(\ell,f) \, U(P)^* \ = \ w (P(\ell,f)) \ ,
\]
for any $|(\ell_1,f_1)|\subseteq o_1$, $|(\ell_2,f_2)| \subseteq  o_2$,
$o_1 \perp o_2 \in K$ and $P \in \sP^\uparrow_+$. 
\end{theorem}

\begin{proof}
By definition, for any closed curve $\ell$  there is a surface $\si$   such that 
$\ell = \partial \si$.
By Lemma \ref{lem.0B:6} we have that $F\left<\si,f\right>$ is independent of the choice of $\si$ with $\partial \si = \ell$, 
so $w$ is well-defined.
Covariance follows by Lemma \ref{0B:5}$(i)$.
To prove causality we consider $|(\ell,f)| \subseteq o$ and note that, since $o$ is convex, we can find $z \in o$ such that 
$h^z(\ell,f)$ is supported within $o$. So the proof follows by Lemma \ref{0B:5}$(iii)$.
\end{proof}

\begin{remark}
\label{0B:7}
Some observations about this theorem are in order. 
\begin{enumerate}
\item Causality of $w$ holds in a more general form. If $(\ell_1,f_1)$ and $(\ell_2,f_2)$ are causally disjoint smearing closed curves not linked together, 
then $[ w(\ell_1,f_1),w(\ell_2,f_2)]=0$. In fact in this case it is always possible to find 
causally disjoint smearing surfaces having $\ell_1$, $\ell_2$ as boundaries, 
\emph{via} suitable cone constructions (see Theorem \ref{0Cc:22} and following remarks). 
\item The above theorem generalizes straightforwardly to generic $1$-cycles $Z_1(\bR^4,\cS)$.
\item The proof of the previous theorem is based on the properties of $F_{\mu\nu}$ 
being strongly causal (\ref{0A:2}), covariant (\ref{0A:1}), and closed as a 2--form (\ref{0A:4}). 
So we do not expect limitations from considering quantum fields with values in a generic Lie algebra,
or defined over generic space-times with trivial de Rham cohomology.
\item Assuming that there is some potential $A$ such that 
$\partial_\mu A_\nu - \partial_\nu A_\mu = F_{\mu\nu}$,
we may apply the Stokes' theorem and write in heuristic terms
\begin{equation}
\label{eq.w}
w(\ell,f) \ = \ \exp \, i \oint_\ell A(f)  \ \ , \ \  (\ell,f) \in \partial\Si_2(\bR^4,\cS) \ , 
\end{equation}
where $A(f)$ is a smearing of $A$. This yields the link with the notion of Wilson loop used in quantum gauge theories.
The rigorous version of the above equality shall be proved in Proposition \ref{prop:0Bb6}.
 
\end{enumerate}
\end{remark}


\subsection{The potential system}
\label{0Bb} 
We have observed that the electromagnetic field defines an exact 2-form of $\bR^4$. 
Using the classical formula for the reconstruction of the primitive of an exact 2-form we shall define the (electromagnetic) potential 1-form,
depending on the origin used to perform the integration; varying the origin in $\bR^4$ gives a system of potential forms that turns out to be causal and covariant.   
We analyze the properties of these operators and define their line integral.

\paragraph{The electromagnetic 1-form.}  We want to study the primitives of the closed 2-form defined by the electromagnetic field.
To this end we may use the cone construction to obtain directly the integral of a primitive on a smearing curve, 
however we prefer to follow  a bottom-up approach first finding the primitive 1-forms. 
The equivalence between these two formulations shall be shown at the end of this section.\smallskip

If  $\omega$ is a closed 2-form on $\bR^4$, then the formula  
\begin{equation}
\label{0Bb:0}
\omega^z_\mu(y) \ = \ \int^1_0  t\,  (y-z)^\alpha \, \omega_{\mu\alpha}(z+t(y-z)) \,  dt \ 
\ \ , \ \
y \in \bR^4 \ ,  
\end{equation}
gives a primitive of $\omega$ for any $z\in\bR^4$, as can be verified by an elementary computation.
Indeed this is the argument used to prove the Poincar\'e Lemma, that is usually presented with
the choice $z=0$. As we shall see soon, the additional 
degree of freedom given by $z \in \bR^4$ is necessary to make explicit the covariance and the causality of the potential.
On this grounds, given a test function $f$ the formal definition of our quantum potential form is 
\[
A^z_\mu(y,f) \ := \ \int^1_0 t(y-z)^\alpha F_{\mu\alpha}(f_{z+t(y-z)}) \, dt 
\ \ , \qquad
y \in \bR^4 \ .
\]
We show below that this formula holds in the sense of sesquilinear forms. 
To give the rigorous definition we observe that for any smearing point $(y,f)\in\Si_0(\bR^4,\cS)$ there corresponds a 1-tensor $y^z_f$, 
defined as   
\begin{equation}
\label{0Bb:1}
y^{z,\alpha}_f(x) \ := \
\int^1_0 (y-z)^\alpha t f_{z+t(y-z)}(x) \, dt    \ ,\qquad  y \in\bR^4 \ ,  
\end{equation}
where, as usual, the components $y^{z,\alpha}_f$ take values  in $\cS(\bR^4)$ and are limit, in the topology of $\cS(\bR^4)$, of the partial Riemann sums defining the integral. If $\supp(f)\subseteq B_\eps(0)$ then 
\begin{equation}
\label{0Bb:2}
y^z_f 
\ \subseteq \ 
B_\eps(r_{z,y}) \ = \ \{ x\in\bR^4 \ | \ d(x,r_{z,y})<\eps\} \ ,
\end{equation}
where $r_{z,y}:=\{ z+t(y-z)\in\bR^4 \ | \ t \in [0,1] \}$.

\begin{definition}
\label{0Bb:3}
For any $z \in \bR^4$, we call the \textbf{(electromagnetic) potential 1-form based on $z$}
the operator-valued map
\begin{equation}
\label{0Bb:3a}
A^z_\mu(y,f) \ := \ F_{\mu\alpha}(y^{z,\alpha}_f) 
\ , \qquad (y,f) \in \Si_0(\bR^4,\cS) \ . 
\end{equation}
The point $z$ is called \textbf{pole}. 
The family $A_\mu = \{ A^z_\mu \}_{z}$ is called the \textbf{potential system}.
\end{definition}
Some basic properties, like the fact that $A^z_\mu(y,f)$ is a weak primitive of the electromagnetic field, 
are shown in the next lemma.
%
\begin{lemma}
\label{0Bb:4}
Let $ (y,f) \in \Si_0(\bR^4,\cS)$ and $z\in\bR^4$. 
Then $A^z_\mu(y,f)$ is an essentially  self-adjoint operator on $\cD$ such that
$A^z_\mu(y,f)\cD\subseteq\cD$.
The mapping  $y\mapsto (\phi,A^z_\mu(y,f)\psi)$ is smooth  and the relations 
\begin{itemize}
 \item[(i)] $(\phi,A^z_\mu(y,f)\psi) \ = \ \int^1_0  t(y-z)^\alpha (\phi,F_{\mu\alpha}(f_{z+t(y-z)})\psi) \, dt$,     
 \item[(ii)] $(\phi,F_{\mu\nu}(f_y)\psi) \ = \ \partial_{y^\mu} (\phi,A^z_\nu(y,f)\psi) - \partial_{y^\nu} (\phi, A^z_\mu(y,f)\psi)$,
\end{itemize}
hold for any $\phi,\psi\in\cD$.
\end{lemma}
\begin{proof}
The facts that $A^z_\mu(y,f)$ is essentially self-adjoint and a linear map of $\cD$ follow by definition.
The equation $(i)$ follows by the same argument used to prove Lemma \ref{0B:5}$(ii)$, that is, 
convergence of the integral defining the form $y^z_f$ in the sense of Schwartz topology.
To prove equation $(ii)$ we define
$\tau_\mu(y) := (\phi,A^z_\mu(y,f)\psi)$, $y \in \bR^4$,
and using $(i)$ we find
\[
\partial_\mu \tau_\nu (y)
\ = \ 
\frac{\partial}{\partial y^\mu} \int^1_0  t(y-z)^\alpha (\phi,F_{\nu\alpha}(f_{z+t(y-z)})\psi) \, dt \ .
\]
The argument used to prove that (\ref{0Bb:0}) is a primitive now applies, 
so $\tau$ is a primitive of the 2-form $y\mapsto (\phi,F_{\mu\nu}(f_y)\psi)$ as desired. 
\end{proof}
We point out two basic features of QED that are reflected on the potential 1-form: 
$A^z_\mu(y,f)$ is neither localized nor covariant with respect to $f$. 
To be more precise $A^z_\mu(y,f)$ is not localized on $\supp(f)$ but, according to 
Definition \ref{0Bb:3}, on the support of the tensor  $y^z_f$, whose localization is estimated by (\ref{0Bb:2}). 
Concerning covariance we have the following result. 
\begin{lemma}
\label{0Bb:5}
For any $ (y,f) \in \Si_0(\bR^4,\cS)$ and $P = (a,L)\in\sP^\uparrow_+$ we have 
\[
U(P) \, A^z_\mu(y,f) \, U(P)^*  \ = \  {{L^{-1}}_\mu}^\delta \, A^{Pz}_{\delta}(P(y,f)) \ . 
\]
\end{lemma}
\begin{proof}
According to (\ref{0A:1}) we have 
\[  
U(P)\, A^z_\mu(y,f) \, U(P)^* \ = \ 
U(P) \, F_{\mu\alpha}(y^{z,\alpha}_f))\,  U(P)^* \ = \
{{L^{-1}}_\mu}^\delta  \,{{L^{-1}}_\alpha}^\beta \, F_{\delta\beta}(y^{z,\alpha}_f\circ P^{-1}) \ .
\]
Using (\ref{0Bb:1}) we compute
\begin{align*}
y^{z,\alpha}_f  & \circ P^{-1}(x)  =\\
& = y^{z,\alpha}_f (L^{-1}(x-a)) 
= \int^1_0 (y-z)^\alpha t \cdot  f( L^{-1}x - L^{-1}a - z -t(y-z) ) \, dt \\
& = \int^1_0 (y-z)^\alpha t \cdot  f_L ( x-a - L(z -t(y-z))) \, dt 
= \int^1_0 (y-z)^\alpha t \cdot  f_L ( x-Lz -t(Ly-Lz) ) \, dt \\
& = {(L^{-1})^\alpha}_\beta \int^1_0 (Py-Pz)^\beta t \cdot  f_L ( x-Pz -t(Py-Pz) ) \, dt \\
& = {(L^{-1})^\alpha}_\beta \,\,  (Py)^{Pz,\alpha}_{f_L}(x) \ .
\end{align*}  
So according to the definition (\ref{0Bb:1}) and  using the relation  ${{L^{-1}}_\alpha}^\beta= {L^\beta}_\alpha$ (see Section \ref{0A}), we have 
\begin{align*} 
U(P)\,  A^z_\mu(y,f) \,U(P)^* 
& = {{L^{-1}_\mu}^\delta} \,{L^{-1}_\alpha}^\beta \, F_{\delta\beta}(y^{z,\alpha}_f  \circ P^{-1}) \\
& = {{L^{-1}_\mu}^\delta} \, {L^{-1}_\alpha}^\beta \, ({L^{-1})^\alpha}_\nu \, 
    F_{\delta\beta}\big( (Py)^{Pz,\alpha}_{f_L}\big) = {{L^{-1}_\mu}^\delta} \,\, F_{\delta\beta}\big( (Py)^{Pz,\alpha}_{f_L}\big)  \\
& = {{L^{-1}_\mu}^\delta} \, A^{Pz}_{\delta}(P(y,f)) \ , 
\end{align*}
completing the proof.
\end{proof}
This result shows that the electromagnetic potential form is not covariant. In particular we see that a Poincar\'e transformation \emph{affects} the pole $z$ with respect to which the potential form is reconstructed. 
In other words, we have covariance only if we consider the potential system $A_\mu$ 
and not a single potential 1-form. We shall return on this point in Section \ref{0C}.

\paragraph{The line integral of the  potential form.} The last step is to analyze the integral of the quantum potential form on a curve $\gamma$. 
The idea is, as usual, to find a suitable test function giving a rigorous sense to the formal expression
\begin{equation}
\label{0Bb:00}
\int_\gamma A^z_\mu(\gamma,f) \, d\gamma^\mu 
\ := \ 
\int^1_0 A^z_\mu(\gamma(s),f) \, \dot{\gamma}^\mu(s) \, ds
\ .
\end{equation}
In analogy to the case of the potential form, we pick a point $z\in\bR^4$ and associate a 2-tensor $\gamma^z_f$ to 
any smearing curve $(\gamma,f)\in\Si_1(\bR^4,\cS)$, by defining 
\begin{equation}
\label{0Bb:10}
\gamma^{z,\mu\nu}_f(x) \ := \ 
\int^1_0 \left(\int^1_0  (\gamma(s)-z)^\mu t f_{z+t(\gamma(s)-z)}(x) dt\right) \dot{\gamma}^\nu(s) \, ds 
\ \ , \ \
x \in \bR^4
\ .
\end{equation}
Again, the coefficients $\gamma^{z,\mu\nu}$ are in $\cS(\bR^4)$ and are limit of the partial Riemann sums defining the integral. Notice that  $\supp(f) \subseteq B_\eps(0)$ implies
\begin{equation}
\label{0Bb:11}
\supp(\gamma^z_f) 
\ \subseteq \
\{ x\in\bR^4 \ | \ d(x,z+t(\gamma(s)-z)) < \eps \ , \  t,s\in[0,1]\}  \ . 
\end{equation}
In words, $\gamma^z_f$ is supported in the envelope between $\supp(f)$ 
and the set spanned by $z+t(\gamma(s)-z)$ as $t,s$ vary in $[0,1]$.\smallskip

After these preliminary observations, the \emph{line integral of the potential form} is the operator
\begin{equation}
\label{0Bb:12}
A^z\left<\gamma,f\right> \ := \ F_{\mu\nu}(\gamma^{z,\mu\nu}_f)
\  , \qquad
(\gamma,f) \in \Si_1(\bR^4,\cS) \ .
\end{equation}
We give some properties of the above line integral. 
We start by observing that $A^z\left<\gamma,f\right>$ is an essentially self-adjoint operator 
defined on $\cD$ such that $A^z\left<\gamma,f\right>\cD\subseteq \cD$. 
The fact that (\ref{0Bb:00}) holds in the weak sense is established in the following result:
\begin{lemma}
\label{lem.0Bb5}
For any $\phi,\psi\in \cD$, $z \in \bR^4$ and $ (\gamma,f) \in \Si_1(\bR^4,\cS)$, we have:
\begin{itemize}
\item[(i)] $(\phi, A^z\left<\gamma,f\right>\psi) \ = \ \int^1_0 (\phi,A^z_\nu(\gamma(s),f)\psi) \, \dot{\gamma}^\nu(s) \, ds$;  
\item[(ii)] $(\phi, A^z\left<\gamma,f\right>\psi) \ = \
\int^1_0\int^1_0  (\gamma(s)-z)^\mu  \, t \, (\phi, F_{\mu\nu}( f_{ z+t(\gamma(s)-z) })\psi)  
                  \dot{\gamma}^\nu(s) \, ds dt$.
\end{itemize}
\end{lemma}
\begin{proof}
$(i)$ We note that $(\gamma(s),f) \in \Si_0(\bR^4,\cS)$ for any $s \in \Delta_1$ and find, according to (\ref{0Bb:1}), that 
$\gamma^{z,\mu\nu}_f(x)  =  \int^1_0\, \gamma(s)^{z,\mu}_f(x) \, \dot{\gamma}^\nu(s) \, ds$ for any $x \in \bR^4$. 
Noticing, as usual, that in this relation  $\gamma^{z,\mu\nu}_f$ is limit in the topology 
of $\cS(\bR^4)$ of the partial Riemann sums defining the integral, we have 
%
\begin{align*}
(\phi, A^z\left<\gamma,f\right>\psi) & =
(\phi,F_{\mu\nu}(\gamma^{z,\mu\nu}_f)\psi) \ = \ 
\int^1_0 (\phi,F_{\mu\nu}\big( \gamma(s)^{z,\mu}_f\big)\psi) \, \dot{\gamma}^\nu(s) \, ds  \\ & \stackrel{(\ref{0Bb:3a})}{=}
\int^1_0 (\phi,A^z_\nu(\gamma(s),f)\psi) \, \dot{\gamma}^\nu(s) \, ds \ , 
\end{align*}
%
as desired. Finally, $(ii)$ follows by Lemma \ref{0Bb:4}$(i)$.
\end{proof}
The definition (\ref{0Bb:12})  makes manifest that the line integral of the potential form is not localized on the support of the test function, but on the subset defined in (\ref{0Bb:11}). The action of the Poincar\'e group is computed in the following result, which, in particular, 
shows how the pole is involved by the transformation:

\begin{proposition}
For any $z \in \bR^4$, $(\gamma,f) \in \Si_1(\bR^4,\cS)$ and $P \in \sP^\uparrow_+$ we have
\begin{equation}
\label{0Bb:14} 
U(P) \, A^z\left<\gamma,f\right> \, U(P)^* \ = \ A^{Pz} \left<P(\gamma,f)\right> \ . 
\end{equation}
\end{proposition}
\begin{proof}
For $P=(x,L)$, we apply Lemma \ref{lem.0Bb5}$(i)$  and Lemma \ref{0Bb:5} obtaining
\begin{align*}
(\phi \, , \, U(P)  A^z\left<\gamma,f\right> U(P)^* \psi ) & =
{{L^{-1}}_\mu}^\delta \int^1_0 (\phi \, , \, A^{Pz}_\delta (P(\gamma(s) , f) )  \psi ) \, \dot{\gamma}^\mu(s) \, ds \\ 
& =
\int^1_0 (\phi \, , \, A^{Pz}_\delta (P(\gamma(s),f))  \psi ) \, \dot{(P \gamma)}^\mu(s) \, ds = 
(\phi \, , \, A^{Pz}\left< P(\gamma,f) \right> \psi ) \ .
\end{align*}
\end{proof}
We now prove the operatorial version of the Stokes' theorem which gives, passing to the exponentials, a rigorous version of (\ref{eq.w}):
\begin{proposition}
\label{prop:0Bb6}
For any smearing closed curve $(\gamma,f)\in \partial\Si_2(\bR^4,\cS)$ and for any $\si \in \Si_2(\bR^4)$ such that $\gamma=\partial \si$, we have
$A^z\left<\gamma,f\right> \ = \ F\left<\si,f\right> $.
\end{proposition}
\begin{proof}
For any $\phi,\psi \in \cD$, we define the forms
\[
\varphi^z_\nu(y) := (\phi,A^z_\nu(y,f),\psi)
\ \ , \ \
\omega_{\mu\nu}(y) := (\phi,F_{\mu\nu}(f_y),\psi)
\ \ , \ \
y \in \bR^4 \ .
\]
Note that by Lemma \ref{0Bb:4}$(ii)$ we have $d \varphi^z = \omega$. 
Applying the Stokes' theorem we find
\begin{align*}
(\phi, A^z\left<\ell,f\right>\psi) & \stackrel{ Lemma\,\ref{lem.0Bb5}(i)}{=} 
\int^1_0 \varphi^z_\nu (\ell(s)) \, \dot{\ell}^\nu(s) \, ds =
\int_\ell \varphi^z  \\
& =  \int_\si \omega  = \int_{\Delta_2} (\phi,F_{\mu\nu}(f_{\si(s)})\psi) \, \si^{\mu\nu}(s) d^2s 
\stackrel{Lemma\, \ref{0B:5}(ii)}{=}  (\phi,F\left<\si,f\right>\psi) \ ,
\end{align*}
and the proof follows by density of $\cD$ in $\cH$.
\end{proof}

We now show the equivalence between the Definition \ref{0Bb:3} of the potential form $A^z_\mu$ and the cone construction. Given a smearing curve $(\gamma, f)\in\Si_1(\bR^4,\cS)$, we consider 
the smearing surface $h^z(\gamma,f):= (h^z\gamma,f)\in\Si_2(\bR^4,\cS)$, where $h^z$ is defined by (\ref{0A:8a}). Explicitly  
\begin{equation}
\label{0Bb:15} 
 h^z\gamma(t,s) \ := \ (t+s)\gamma\left(\frac{s}{t+s}\right) +(1-t-s)z \ , \qquad (t,s) \in\Delta_2 \ ,  
\end{equation}
and, as  can be verified by a direct calculation,  the surface element  $(h^z\gamma)^{\mu\nu}(t,s)\,ds\,dt$ is given by 
\begin{equation}
\label{0Bb:15a} 
(h^z\gamma)^{\mu\nu}(t,s) = \left(\gamma\left(\frac{s}{t+s}\right)-z\right)^\mu \dot\gamma^\nu\left(\frac{s}{t+s}\right)-  \left(\gamma\left(\frac{s}{t+s}\right) -z\right)^\nu \dot\gamma^\mu\left(\frac{s}{t+s}\right)\ .
\end{equation}
Then it holds
\begin{proposition} 
\label{0Bb:16}
For any $z\in\bR^4$ and any smearing curve $(\gamma,f)\in\Si_1(\bR^4,\cS)$ we have 
\[
 A^z\left<\gamma,f\right> \ = \ F\left<h^z\gamma,f\right> \ . 
\]
\end{proposition}
\begin{proof}
Let $\gamma^{z,\mu\nu}_f$ be the 2-tensor defined by (\ref{0Bb:10}). We change the variable as $s_1:= ts$ and $t_1:=t(1-s)$. Then 
$t=t_1+s_1$, $s=s_1/(t_1+s_1)$ and the Jacobian of this transformation equals $1/(t_1+s_1)$. Hence   
\begin{align*}
\gamma^{z,\mu\nu}_f(x) & = \int_{\Delta_2} \left(\gamma^\mu\left(\frac{s_1}{t_1+s_1}\right)- z^\mu\right)\,  \dot\gamma^\nu\left(\frac{s_1}{t_1+s_1}\right) \,
f_{z+ \left(t_1+s_1\left(\gamma(\frac{s_1}{s_1+t_1}\right) -z\right)}(x) 
 \, ds_1 dt_1 \\ 
 & = \int_{\Delta_2} \left(\gamma\left(\frac{s_1}{t_1+s_1}\right)- z\right)^\mu\, \dot \gamma^\nu\left(\frac{s_1}{t_1+s_1}\right)\,  f_{h^z\gamma(t_1,s_1)}(x) \,  ds_1 dt_1 \ .
\end{align*}
Recalling (\ref{0A:13}) we get 
\[
\gamma^{z,\mu\nu}_f(x) - \gamma^{z,\nu\mu}_f(x) \ = \
\int_{\Delta_2}  f_{h^z\gamma(t_1,s_1)}(x)  (h^z\gamma)^{\mu\nu}(t_1,s_1) \,  ds_1 dt_1 \ = \
h^z\gamma[f]^{\mu\nu}(x)
\ ,
\]
and, by antisymmetry of $F_{\mu\nu}$, we obtain
\[
A^z\left<\gamma,f\right> \ = \
F_{\mu\nu}(\gamma^{z,\mu\nu}_f) \ = \
\frac{1}{2} F_{\mu\nu}(\gamma^{z,\mu\nu}_f - \gamma^{z,\nu\mu}_f) \ = \
\frac{1}{2} F_{\mu\nu}(h^z\gamma[f]^{\mu\nu})= F\left<h^z\gamma,f\right> \, ,
\]
completing the proof.
\end{proof}
Note that when $\gamma$ is closed, Proposition \ref{0Bb:16} is a particular case of Proposition \ref{prop:0Bb6}.

\subsection{Gauge transformations}
\label{0Bc}
The notion of local gauge transformation involves a gauge function $\rg \in C^2(\bR^4)$ whose gradient is added to the potential,
$A_\mu \mapsto A_\mu + \partial_\mu \rg$.
Coherently with the previous subsections, we apply a smearing $f  \in \cS(\bR^4)$ and add the variable $y \in \bR^4$ in terms of a translation,
so we define, for any function $\rg \in C^2(\bR^4)$ with bounded gradient,
\[
A^z_\mu (y,f) \mapsto A^{\rg,z}_\mu (y,f) \ := \ A^z_\mu (y,f) + \partial_{y^\mu} \rg(y,f) 
\ \ , \ \
(y,f) \in \Si_0(\bR^4,\cS)
\ ,
\] 
where, for any bounded $\rh \in C^2(\bR^4)$, we used the notation
\[
\rh(y,f)    \ := \   \int \rh(y-x)f(x) \, d^4x \ \ , \ \ y \in \bR^4 \ , 
\]
for the convolution.
This ensures that the gauge transformation $\rg$ leaves the observable $F_{\mu\nu}$ unaffected, as follows applying Lemma \ref{0Bb:4}$(ii)$
\[
  \partial_{y^\nu} A^{\rg,z}_\mu(y,f) -  \partial_{y^\mu} A^{\rg,z}_\nu (y,f)
  \ = \
  F_{\mu\nu}(f_y) +  ( \partial_{y^\mu} \partial_{y^\nu} -  \partial_{y^\nu} \partial_{y^\mu} ) \rg (y,f)  
  \ = \ 
  F_{\mu\nu}(f_y) \ ,
\]
so the argument for proving Proposition \ref{prop:0Bb6} yields
\[
A^{\rg,z}\left<\ell,f\right> \ = \ A^z\left<\ell,f\right> \ , 
\]
for any closed smearing curve $(\ell,f) \in \partial\Si_2(\bR^4,\cS)$. 

\paragraph{Covariance.}
One can easily verify that the system $A^{\rg,z}_\mu (y,f)$, $z\in\bR^4$, 
does not transform according to Lemma \ref{0Bb:5}. 
Since we want a gauge transformation both leaving $F_{\mu\nu}$ unaffected and such that 
$A^\rg$ is covariant in the sense of  Lemma \ref{0Bb:5},
we consider a family of gauge functions, also denoted by $\rg$, and we write $\rg = \{ \rg^z \}_{z \in \bR^4}$ where $\rg^z \in C^2(\bR^4)$, $z \in \bR^4$, having bounded gradient
and fulfilling the covariance property
\begin{equation}
\label{eq.gt1}
\rg^{Pz}(y,f) \ = \ \rg^z(P(y,f))\ , \quad     (y,f) \in \Si_0(\bR^4,\cS) \ , \ P \in \sP^\uparrow_+ \ .
\end{equation}
We denote the set of functions fulfilling the above conditions by $\cG$, which is clearly an Abelian group under
the operation of pointwise sum. Writing the covariance condition in explicit terms we find, 
with $P=(a,L)$,
\begin{align*}
\int \rg^{Pz}(y-x)f(x) \, d^4x & =
\int \rg^z(Py-x)f_L(x) \, d^4x \\
& =  
\int \rg^z(Py-Lx')f(x') \, d^4(Lx') =\int \rg^z(Py-Lx')f(x') \, d^4x' \ .
\end{align*}
Since the above equality must hold for any $f \in \cS(\bR^4)$, we conclude that (\ref{eq.gt1}) is equivalent to
require
$\rg^{Pz}(y-x) = \rg^z(a+Ly-Lx)$, for any $x,y \in \bR^4$,
that is,
\begin{equation}
\label{eq.gt2}
\rg^{Pz} \ = \ \rg^z \circ P \ \ \ , \ \ \    z\in\bR^4 \ , \ P \in \sP^\uparrow_+ \ .
\end{equation}
The potential system transforms as  
\begin{equation}
 \label{0Bb:7}
 A^{\rg,z}_\mu (y,f) \ := \ A^z_\mu (y,f) + \partial_\mu  \rg^z (y,f)  \ , \qquad \rg\in\cG \ ,
 \end{equation}
in such a way that
\begin{equation}
 \label{0Bb:8}
 \partial_{y^\nu} A^{\rg,z}_\mu (y,f) - \partial_\mu A^{\rg,z}_\nu (y,f) \ = \ F_{\mu\nu}(f_y)
 \ \ \ , \ \ \
 A^{\rg,z}\left<\ell,f\right> \ = \ A^z\left<\ell,f\right> \ ,
\end{equation}
for any $(\ell,f) \in \partial\Si_2(\bR^4,\cS)$. Passing to covariance, we compute
\[
 \partial_\mu \rg^z (y,f) \ = \  
 \partial_\mu \rg^{Pz}(P(y,f)) \ = \
 {{L^{-1}}_\mu}^\nu  \, \partial_{(Py)^\nu} \rg^{Pz}(P(y,f)) \ ;
\]
this relation, (\ref{0Bb:7}) and Lemma \ref{0Bb:5} give the desired covariance
\begin{align*}
U(P) A^{\mathrm{g},z}_\mu(y,f) U(P)^* & = 
U(P) A^z_\mu(y,f) U(P)^*  +  \partial_\mu \rg^z (y,f) \\ 
& =
{{L^{-1}}_\mu}^\delta \, A^{Pz}_{\delta}(P(y,f)) + {{L^{-1}}_\mu}^\delta \, \partial_{(Py)^\delta } \rg^{Pz}(P(y,f))  \\
& =
{{L^{-1}}_\mu}^\delta \,  A^{\rg, Pz}_\delta(P(y,f)) \ .
\end{align*}
An example of a gauge transformation can be easily given by taking $g \in C^2(\bR)$ and defining
$\rg^z (y) := g((y-z)^2)$, where $(y-z)^2 = (y-z)\cdot (y-z)$.

\paragraph{Gauge transformations of the line integral.} Let $(\gamma,f) \in \Si_1(\bR^4,\cS)$. We compute
\begin{align*}
(\phi, A^{\rg,z}\left<\gamma,f\right>\psi) & =
\int^1_0 (\phi,A^z_\nu(\gamma(s),f)\psi) \, \dot{\gamma}^\nu(s) \, ds \ + \
(\phi,\psi) \int^1_0 \left\{ \partial_{y^\nu} \rg^z(y,f)\right\}_{y=\gamma(s)} \dot{\gamma}^\nu(s) \, ds \\ 
& = (\phi, A^z\left<\gamma,f\right>\psi) + (\phi,\psi) \, \{ \rg^z(\gamma(1),f) - \rg(\gamma(0),f) \} \ ,
\end{align*}
and find the familiar expression
\[
A^{\rg,z}\left<\gamma,f\right> \ = \ A^z\left<\gamma,f\right> + \rg^z(\gamma(1),f) - \rg^z(\gamma(0),f) \ ,
\]
which shows the appearance of the scalar factors $\rg^z(\gamma(t),f)$, $t=0,1$.
So when $\gamma(0)=\gamma(1)$ we find 
$A^{\rg,z}\left<\gamma,f\right> = A^z\left<\gamma,f\right>$.


\section{Representations induced by the electromagnetic field}
\label{0C}

In this section we show that the electromagnetic field induces a representation of the net of causal loops over the Minkowski space-time. 
Since we shall use as generators of loops \emph{smearing, affine} 1-simplices, our causal loop net is different from that defined for an arbitrary globally hyperbolic space-time \cite{CRV}, anyway the construction is the same. \\
\indent Proceeding as in \cite{CRV}, we show that representations of the net of causal loops have a geometrical interpretation in terms of connection systems of $\Si_*(\bR^4,\cS)$, and this leads to a natural notion of gauge transformation. The feature introduced in the present paper is that representations can be also equivalently described in terms of 2-cochains of $\Si_*(\bR^4,\cS)$, as a consequence of the fact that $\bR^4$ is contractible. \\
\indent This is a key observation: in fact we shall show that $F_{\mu\nu}$ defines, \emph{via} the associated 2-form, 
a 2-cochain $\rw^{em}$, hence a representation of the net of causal loops. 
Using the abstract procedure outlined above, we consider the connection system $\ru^{em}$ defined by $\rw^{em}$, 
and show that $\ru^{em}$ is equivalent to the potential system $A_\mu$ introduced in the previous section.
Finally, gauge transformations of $A_\mu$ define gauge transformations of $\ru^{em}$.


\subsection{The simplicial set and the set of paths} 
\label{0Ca}

As already mentioned we modify the definition of the simplicial set made in \cite{CRV}, considering the affine subcomplex of $\Si_*(\bR^4,\cS)$.
This simplicial set is not pathwise connected, however all the key constructions made in \cite{CRV}  apply, 
that is, covariance under the action of the Poincar\'e group and the notion of covariant path-frame.

\paragraph{The simplicial set $\Si_*$.}  The corner stone of the net of causal loops defined in \cite{CRV} is the simplicial set denoted here by $\Si_*(K)$,
defined in terms of double cones of the Minkowski space-time and their inclusions.  This choice encodes the essential properties of localization and covariance  
that reflect into the net of causal loops. Namely, we defined a non-Abelian free group generated by 1-simplices of  $\Si_*(K)$, then   
we considered the subgroups generated by loops (suitable compositions of 1-simplices, see next paragraph), 
and then we used these groups to define the net.\\
\indent Now, we want a simplicial set encoding localization and covariance, and able to deal with the smearing of quantum fields. 
To this end, the natural choice is the simplest subsimplicial set of $\Si_*(\bR^4,\cS)$, that of affine singular smearing simplices. 
This results to be the natural choice because of the vanishing of  the integral over  \emph{non-injective} affine simplices, namely  degenerated simplices  defined below. \smallskip 

An \emph{affine singular smearing} $n$-simplex $s$ is a pair $(\varphi_s,f)$ where $\varphi_s:\Delta_n\to \bR^4$ is an affine function  and $f\in\cS$. 
We denote the set of affine singular $n$-simplices by $\Si_n$, and the corresponding simplicial set by $\Si_*$.  \\ 
\indent An affine smearing $n$-simplex $s$ is determined by the knowledge of its vertices. In fact defining  
$s_i := \varphi_s(e_i)$, 
where each $e_i$, $i=0,\ldots,n$, is a vertex of $\Delta_n$, we have
\begin{equation}
\label{0Ca:0}
\varphi_s(t) \ = \ s_0 + \sum^{n}_{i=1}t_i(s_i-s_0)
\ \ , \qquad
t \in \Delta_n \ .
\end{equation}
So we write 
\begin{equation}
\label{0Ca:1}
s = (s_0,\ldots, s_n;f) \ \ , \qquad  s \in \Si_n \ ;
\end{equation}
a smearing  affine $n$-simplex $s$ is said to be \emph{degenerated} if two of its vertices coincide.\\ 
\indent We shall denote 0-, 1-, 2-simplices by the letters $a,b$ and $c$ respectively, and observe that $0$-simplices are points, 1-simplices are segments, and  2-simplices are triangles. 
In particular, we denote the parametric function (\ref{0Ca:0}) of $b \in \Si_1$ by $r_b:\Delta_1\to \bR^4$ and that of $c \in \Si_2$ by $\si_c:\Delta_2\to \bR^4$.
Concerning 1-simplices,  
if $b=(b_0,b_1;f)$ then $\partial_0b=(b_1;f)$ and $\partial_1b=(b_0;f)$. So the 1-face corresponds to the 0-vertex and conversely. Instead,  
the faces of a 2-simplex $c=(c_0,c_1,c_2;f)$ are given by $\partial_0 c=(c_1,c_2;f)$, $\partial_1 c=(c_0,c_2;f)$ and $\partial_2 c=(c_0,c_1;f)$. 
According to the orientation  described in Subsection \ref{0Ab},  
the \emph{opposite} $\bar{b}$ of $b$ is the 1-simplex $\bar{b}=(b_1,b_0;f)$ or, equivalently in terms of face relations, $\partial_0\bar b= \partial_1 b$ and $\partial_1\bar b=\partial_0 b$. 
Instead,  the \emph{opposite} of $c$  is the 2-simplex $\bar c=(c_0,c_2,c_1;f)$ or, in terms of face relations,  
$\partial_0 \bar{c}= \overline{\partial_0 c}$, $\partial_1 \bar{c}= \partial_2  c$ and $\partial_2 \bar{c}= \partial_1 c$.

\paragraph{Words, paths, loops, and path-frames.} We use $\Si_1$ as an alphabet for generating words and, in particular, paths. A finite ordered sequence  $w=b_n b_{n-1}\cdots b_1$ of 1-simplices is called a \emph{word}. The \emph{opposite} of a word  $w=b_n\cdots b_1$ is the word $\overline{w}:=\overline{b}_1,\ldots,\overline{b}_n$. We shall denote the \emph{empty} word by $\mathbbm{1}$. 
We have to care of not confusing these words with the elements of $C_1(\bR^4,\cS)$: 
actually, we shall see that these words define a \emph{non-Abelian} free group.

The action of the Poincar\'e group extends from 1-simplices, see (\ref{0A:12}), to words, by  
\begin{equation}
\label{0Ca:4a}
Pw \ := \ P b_n\, \cdots \,Pb_1 \ , \ \qquad   P \in \sP^\uparrow_+ \ , 
\end{equation}
and $P\mathbbm{1}:=\mathbbm{1}$ on the empty word.
A word $w=b_n\ldots b_1$ is said to be a \emph{path} whenever its generators satisfy the relation
\begin{equation}
\partial_{0}b_{i+1}=\partial_1b_{i} \ , \qquad i= 1,\ldots, n-1\,.
\end{equation}  
Note that this implies that the generators of a path have all the same smearing function. 
We set $\partial_1w:=\partial_1b_1$ and  $\partial_0w:=\partial_0b_n$ and call 
these 0-simplices, respectively, the \emph{starting} and the \emph{ending point} of the path $w$. 
We shall also use the notation 
\[
w:a\to o
\]
to denote a path from $a$ to $o$.  Given  two paths $p:o\to o'$ and $q:o'\to o''$, since $\partial_0p=o'=\partial_1q$  
the juxtaposition $pq$  of the generators of $q$ an $p$ gives  a new path $qp:o\to o''$, called the \emph{composition}. 
Finally, a path $w:o\to o$ is said to be a \emph{loop over $o$}.\\ 
\indent The boundary of $c \in \Si_2$ in the sense of homology is the 1-cycle $\partial_0c - \partial_1c +\partial_2c$. 
On the other side, we have a geometric notion of boundary, defined as the loop
\begin{equation}
\label{0Ca:4b}
\boldsymbol{\partial} c \, : \, \partial_1 \partial_2 c\to \partial_1 \partial_2 c
\ \ , \ \
\boldsymbol{\partial} c \ := \ \overline{\partial_1 c} \, \partial_0 c \, \partial_2 c \ , \qquad c\in\Si_2 \ .  
\end{equation}
We call $\boldsymbol{\partial} c$ the \emph{path-boundary} of $c$. We stress the difference between the path-boundary $\boldsymbol{\partial} c$ and the boundary $\partial c$ of a 2-simplex $c$: the former is a word,  i.e.\  an \emph{ordered sequence} of 1-simplices.
Note that by the convention adopted in the definition of the opposite $\bar c$ of a $2$-simplex $c$, we have
$\boldsymbol{\partial} \bar c  =  
\overline{\partial_1 \bar c}\, \partial_0 \bar c\, \partial_2 \bar c  =       
\overline{\partial_2 c}\, \overline{\partial_0 c}\, \partial_1 c = \overline{\boldsymbol{\partial} c}$. \\
\indent The simplicial set $\Si_*$ is not pathwise connected since no path joins two 0-simplices having different smearing functions; actually, the connected components are indexed by the smearing function itself. So, given $a=(a_0;f) \in \Si_0$  we define
\[
\Si_n^a \ := \ \{ s \in \Si_n \ : \ s=(s_0,\ldots,s_n;f) \} \ , \qquad n \in \bN \ .
\]
If $a'\in \Si_0^a$, then the 1-simplex $e_{(a',a)}$ defined by 
\begin{equation}
\label{0Ca:5}
\partial_1e_{(a',a)} \ := \ a \ \  , \ \ \partial_0e_{(a',a)}= a' \ , 
\end{equation}
connects $a$ to  $a'$ and is written explicitly as the segment
$e_{(a',a)} = (\varphi_{(a',a)},f)$, $\varphi_{(a',a)}(t) = a_0 + ta'_0$, $t \in \Delta_1$.
In particular, we use the notation 
\begin{equation}
\label{0Ca:5a}
 e_a \ := \ e_{(a,a)} \ \ , \qquad a\in\Si_0 \ ,  
\end{equation}
which is the degenerated 1-simplex (see Subsection \ref{0Ca}) whose vertices equals $a_0$, and having smearing function $f$. We call  $e_a$ the \emph{trivial loop} over $a$. \\
\indent A \emph{path-frame} over a pole $a \in \Si_0$ is a set of paths 
\[
\cP_a \ = \ \{ p_{(a,a')}:a'\to a \ , \ a'\in\Si^a_0 \} \ ,
\]
satisfying the condition that $p_{(a,a)}=e_a$. A \emph{covariant path-frame  system} 
is a collection of path-frames $\cP = \{ \cP_a , a \in \Si_0 \}$ satisfying the relation 
$P\cP_a= \cP_{Pa}$, i.e.\ $Pp_{(a,a')} = p_{Pa,Pa'}$ for any $P\in \sP^\uparrow_+$. 
Notice that, if we set $\cE_a:=\{ e_{(a,a')} \ | \ a'\in\Si^a_0\}$ then the collection $\cE:=\{ \cE_a \ | \ a\in\Si_0\}$ is a covariant path-frame system. We shall refer to $\cE$ as the \emph{Euclidean} path-frame system.

\subsection{The net of causal loops}
\label{0Cb}

We now follow the route of \cite{CRV} and define the net of causal loops using the simplicial set $\Si_*$. Since in \cite{CRV} we started from a different simplicial set the resulting net is different, nevertheless the procedure is the same up to the following (not substantial) points.

First, here we prefer to use also degenerate simplices, because this will simplify the construction of the representation induced by the electromagnetic field. The definition of net of causal loops remains unchanged, because degenerate 1-simplices turn out to be equivalent to the identity of the group. 

Second, we omit the functorial picture of the net with respect to space-times, since we are interested in the case of the Minkowski space-time only.

\paragraph{Groups of loops.} Our construction starts defining the group $\rF$ generated by $\Si_1$ with relations 
\begin{equation}
\label{0Cb:1}
b\overline{b}=\mathbbm{1} \ \ , \ \ e_a = \mathbbm{1}
\ \ , \ \quad
b \in \Si_1 \ , \  a \in \Si_0 \ ,
\end{equation}
where $\mathbbm{1}$ is the identity. These two relations are equivalent to 
\[
b^{-1} = \bar b \ \ , \ \ e_ae_a=e_a \ \  , \ \qquad  b \in \Si_1 \ , \ a \in \Si_0 \ .
\]
The group $\rF$ is (non-canonically) isomorphic to a free group whose elements are \emph{reduced words}, that is, words $w$ in which either pairs of the form $b\bar b$ or degenerated $1$-simplices $e_a$ do not appear. Hence any word $w$ is associated to a unique \emph{reduced} word denoted by $w^\rr$. Finally, we observe that if a word is a path $w:a\to a'$,
then its reduced word is still a path $w^\rr:a\to a'$. \\
\indent The notion of support of a $1$-simplex extends to elements of $\rF$ as follows: the \emph{support} $|w|$ of $w \in \rF$ is the subset of $\bR^4$ obtained as  the union of the supports of the generators of the reduced word $w^\rr$.
For instance, if $w = b_2 b\, \overline{b}\, b_1 e_a $  with  $b_1 \neq \overline{b}_2$, then 
$|w|=|b_2|\cup |b_1|$ since $w^{\rr} = b_2 b_1$.\\
\indent Observing that $P\bar b= \overline{Pb}$ for any $b \in \Si_1$, $P \in \sP^\uparrow_+$, we have that (\ref{0Ca:4a}) defines an action of the Poincar\'e group on $\rF$. This action sends reduced words into reduced words, and this implies that
\begin{equation} 
\label{0Cb:2}
P|w| \ = \ x+L|w| \ = \ |Pw| 
\ \ , \ \   P = (x,L) \in \sP^\uparrow_+ \ , \ w \in \rF \ .
\end{equation}
Note that if $p:a\to a'$, then $Pp : Pa \to Pa'$.
\begin{definition}
\label{0Cb:3}
We call \textbf{group of loops} the subgroup $\rL(\bR^4)$ of $\rF$ generated by loops.
\end{definition}
An element of $\rL(\bR^4)$ is, by definition, a (reduced) word of the form $w=p_n\,p_{n-1}\,\cdots\, p_1$, 
where any $p_i$ is a loop over $a_i \in \Si_0$. It is easily seen that $\rL(\bR^4)$ is stable under reduction of words. Furthermore, $\rL(\bR^4)$ inherits from $\rF$ the Poincar\'e action.

\paragraph{Causal loops.} Using the group of loops we construct a causal net of groups over the set of double cones of the Minkowski space-time. To begin with, we define  
\begin{equation}
\label{0Cb:5}
\rL_o:= \{ w\in\rL(\bR^4) \ | \ |w|\subseteq o\} \ , \qquad o\in K \ , 
\end{equation}
and observe that, according to the definition of the support of an element of $\rL(\bR^4)$, this definition refers not to a generic word $w$ but to its reduced $w^\rr$. 
The set $\rL_o$ is a subgroup of $\rL(\bR^4)$: in fact if $|w|,|w'|\subseteq o$ then 
$|\bar w|\subseteq o$ and $|ww'|\subseteq o'$, moreover $\mathbbm{1}\in\rL_o$ for any double cone $o$ by (\ref{0Cb:1}).
Since $\rL_{o}\subseteq\rL_{o'}$, $o\subseteq o'$, the mapping $o\mapsto \rL_o$ forms a net of subgroups of $\rL(\bR^4)$ over the set of double cones. By (\ref{0Cb:2}) this net is covariant, i.e.\ $P \rL_{o} = \rL_{Po}$, $P \in \sP^\uparrow_+$.\\

We now  impose on $\rL(\bR^4)$ the relations
\begin{equation}
\label{eq.c}
w_1w_2 = w_2w_1  \ \ , \ \ w_i\in\rL_{o_i} \ , \ i=1,2  \ , \ o_1\perp o_2 \ ,
\end{equation}
defining the \emph{group of causal loops}, that we denote by $\widehat{\rL}(\bR^4)$.
We then have an induced net
\begin{equation}
\label{0Cb:7}
\widehat{\rL}_o \ \subseteq \ \widehat{\rL}(\bR^4) \ \ , \ \qquad o\in K \ ,
\end{equation}
where each $\widehat{\rL}_o$ is the subgroup generated by the image of $\rL_o$ under the quotient
defined by (\ref{eq.c}).
Since Poincar\'e transformations preserve the causal disjointness relation of subsets of $\bR^4$,
we have that (\ref{0Ca:4a}) induces the action
\begin{equation}
\label{0Cb:8}
w \mapsto \beta_P(w) \in \widehat{\rL}_{Po} 
\ \ , \ \qquad 
o \in K \ , \ w \in \widehat{\rL}_o \ , \  P \in \sP^\uparrow_+ \ .
\end{equation}
Clearly, the net $\widehat{\rL} = \{ \widehat{\rL}_o \}_{o \in K}$ is causal by (\ref{eq.c}).

\paragraph{The net of causal loops.}  Let $\rC^*$ denote the functor assigning the group $\rC^*$-algebra $\rC^*G$ to the locally compact group $G$. If $G$ is discrete then $\rC^*G$ is unital. Furthermore, if  $G_1$, and $G_2$ are discrete  and $\rho:G_1\to G_2$ is an injective group morphism, then $\rC^*\rho : \rC^*G_1 \to \rC^*G_2$ 
is a unital, injective $^*$-morphism. \smallskip

We now come to the definition of the net of causal loops. We first consider the $\rC^*$-algebras
$\rC^*\widehat \rL(\bR^4)$ and $\rC^*\widehat \rL_o$, $o\in K$. 
Since $\widehat{\rL}_o\subseteq \widehat{\rL}_{o'}$ for any $o\subseteq o'$, 
by functoriality there is a unital, injective $^*$-morphism
$\jmath_{o'o} : \rC^*\widehat \rL_o \to \rC^*\widehat \rL_{o'}$ satisfying $\jmath_{o''o'}\circ\jmath_{o'o}=\jmath_{o''o}$ for any $o\subseteq o'\subseteq o''$. 
For the same reason we have unital, injective $^*$-morphisms 
$\vec{\jmath}_o :\rC^* \widehat \rL_o \to \rC^* \widehat \rL(\bR^4)$ for any $o \in K$, 
such that 
$\vec{\jmath}_{o'} \circ \jmath_{o'o} = \vec{\jmath}_o$ for any $o\subseteq o'$.
On these grounds:
\begin{definition}
\label{0Cb:9}
Let $\cA$ denote the mapping $\cA:K\ni o \to \cA_o\subseteq \cA(\bR^4)$, where 
\[
\cA(\bR^4) := \rC^* \widehat \rL(\bR^4) 
\ \ , \ \ 
\cA_o := \vec{\jmath}_o ( \rC^* \widehat \rL_o ) \ , \ o\in K \ , 
\]
and let $\alpha : \sP^\uparrow_+ \to {\bf aut}\cA(\bR^4)$ be the action of the Poincar\'e group
defined by applying the functor $\rC^*$ to (\ref{0Cb:8}).
We call the pair $(\cA,\alpha)$ \textbf{the net of causal loops} over $K$. 
\end{definition} 
The map $\cA$ is clearly a net, which is causal by (\ref{eq.c}) and covariant by (\ref{0Cb:8}).
We conclude by noting that $\cA$ is not trivial, and that the $\rC^*$-algebras $\cA_o$ are non-Abelian (see \cite{CRV}).

\subsection{Representations}
\label{Cc}

We now discuss representations of $(\cA,\alpha)_K$ and point out their geometrical meaning.

We start by recalling that these representations are in 1-1 correspondence with a particular class of representations of the group of loops $\rL(\bR^4)$. This equivalence yields a geometrical interpretation of representations in terms of causal and covariant connection systems. Using the fact that $\bR^4$ is contractible, we also show that representations of $(\cA,\alpha)_K$ can be equivalently described by a class of 2-cochains of $\Si_*$. 
This is a key result: in fact, we shall use the latter equivalence to prove that the electromagnetic field induces, \emph{via} integration on 2-simplices, a representation of $(\cA,\alpha)_K$ (see Subsection \ref{0Ce}).

\paragraph{Representations and connection systems.}  A covariant representation of the net of causal loops $(\cA,\alpha)_K$ is a pair $(\pi,U)$, where $\pi : \cA(\bR^4) \to \cB\cH$ is a (non degenerated) representation and 
$U : \sP^\uparrow_+ \to \cU\cH$ is such that 
\[
\ad_{U(P)} \circ \pi \ = \ \pi\circ\alpha_P \ \ , \ \qquad P \in  \sP^\uparrow_+ \ . 
\]
These representations are in 1-1 correspondence with \emph{causal and covariant representations $(\lambda,U)$ of the group of loops $\rL(\bR^4)$}: that is, we have unitary representations 
\[
\lambda : \rL(\bR^4) \to \cU\cH \ \ , \ \ U : \sP^\uparrow_+ \to \cU\cH \ ,
\]
satisfying the following properties:
\begin{itemize}
\item[(a)] $[\lambda(p), \lambda(q)]=0$ for loops $p\in L_{o_1}$, $q\in L_{o_2}$ with $o_1\perp o_2$;
\item[(b)] $\ad_{U(P)} \circ \lambda = \lambda \circ P$ for any $P \in \bar \sP^\uparrow_+$. 
\end{itemize}
We refer the reader to \cite{CRV} for the (easy) proof of the above 1-1 correspondence. 
Instead we focus on the relation between representations and connection systems.

A \emph{connection system}
\footnote{The definition of connection system given here generalizes that introduced in \cite{CRV} to a non-pathwise  connected simplicial set. One can easily see that this coincides with the previous one when restricted to any connected component.} 
is a family $\ru$ of maps 
\[
\ru_a : \Si_1^a \to \cU\cH   \ \ , \ \qquad  a \in \Si_0 \ , 
\] 
satisfying
\begin{equation}
\label{0Cc:1}
\ru_a(\bar b) = \ru_a(b)^* \ \ , \ \  \ru_a(e_a)=\mathbbm{1} 
\ \ , \ \qquad a \in \Si_0 \ , \ b \in \Si_1^a \ .  
\end{equation}
We extend $\ru$ to paths as follows: for any $0$-simplex $a$ and any path $p=b_n\cdots b_1$ with $b_i\in\Si_1^a$ for any $i$, we define 
\begin{equation}
\label{0Cc:2}
\ru_a(p) \ := \ \ru_a(b_n)\cdots \ru_a(b_2)\ru_a(b_1) \ . 
\end{equation}
We say that $\ru$ is \emph{causal} whenever for any pair of loops $ p:a\to a$ and $p':a'\to a'$ such that   
 $p\in\rL_{o}$ and $p'\in\rL_{o'}$  with $o\perp o'$, we have 
 \begin{equation}
 \label{0Cc:3}
 [\ru_a(p), \ru_{a'}(p')]=0  \ .
 \end{equation}
Notice that, in general, $p$ and $p'$ are loops in different connected components, i.e.\ they may have different smearing functions. We say that $\ru$ is \emph{covariant} whenever there is a unitary representation 
$U : \sP^\uparrow_+ \to \cU\cH$
such that 
\begin{equation}
\label{0Cc:4}
\ad_{U(P)} \circ \ru_a \ = \ \ru_{Pa} \circ P  \ \ , \ \qquad P\in\sP^\uparrow_+ \ .  
\end{equation}
Covariant connection systems $(\ru,U)$ and $(\ru',U')$  are \emph{equivalent} if there is a family of unitary mappings $\rt=\{\rt_a:\Si^a_0\to \cU(\cH,\cH') \ , \ a\in\Si_0\}$ satisfying 
\begin{equation}
\label{0Cc:5}
\rt_a(\partial_0 b) \, \ru_a(b) \ = \ \ru'_a(b) \, \rt_a(\partial_0 b) 
\ \ \ , \ \ \
\rt_a(a') \, U(P) \ = \ U'(P) \, \rt_{Pa}(Pa') \ , 
\end{equation}
for any $a\in\Si_0$, $b\in\Si_1^a$, $a'\in\Si^a_0$ and $P\in  \sP^\uparrow_+$.

\

\begin{lemma}
\label{0Cc:6}
There exists, up to equivalence, a 1-1 correspondence between causal and covariant representations of the group of loops $\rL(\bR^4)$ and causal and covariant connection systems.  
\end{lemma}
\begin{proof}
Let $(\lambda,U)$ be a causal and covariant representation of $\rL(\bR^4)$ and $\cP$ a path-frame system.
We define, for any path-frame $\cP_a$, 
\begin{equation}
\label{0Cc:7}
\ru^\lambda_{\cP_a}(b)
\ := \ 
\lambda \left(\overline{p}_{(a,\partial_0b)}  
\, 
b  \, p_{(a,\partial_1b)} \right) \ , \qquad b\in\Si_1^a \ ,  
\end{equation}
where $\overline{p}_{(a,\partial_0b)}$ denotes the opposite of  $p_{(a,\partial_0b)}$.  
By definition we have 
$\ru^{\lambda}_{\cP_a}(\bar b) = \ru^{\lambda}_{\cP_a}(b)^*$ and $\ru^{\lambda}_{\cP_a}(e_a) = {\mathbbm{1}}$.   
If $p$ is a loop over $a$, $p=b_n\cdots b_1$, then 
\begin{align*}
\ru^{\lambda}_{\cP_a}(p) & = \ru^{\lambda}_{\cP_a}(b_n)\cdots \ru^{\lambda}_{\cP_a}(b_2) \ru^{\lambda}_{\cP_a}(b_1) \\
& =  {\lambda}\left(\overline{p}_{(a,\partial_0b_n)}  \, b_n  \, p_{(a,\partial_1b_n)} \right)  \cdots 
{\lambda}\left(\overline{p}_{(a,\partial_0b_2)}  \, b_2  \, p_{(a,\partial_1b_2)} \right)  
{\lambda}\left(\overline{p}_{(a,\partial_0b_1)}  \, b_1  \, p_{(a,\partial_1b_1)} \right) \\
& =  {\lambda}\left(\overline{p}_{(a,\partial_0b_n)}  \, b_n  \, p_{(a,\partial_1b_n)} \,   \cdots \overline{p}_{(a,\partial_0b_2)}  \, b_2  \, p_{(a,\partial_1b_2)}\, 
  \overline{p}_{(a,\partial_0b_1)}  \, b_1  \, p_{(a,\partial_1b_1)} \right) \\
& =  {\lambda}(p) \ , 
\end{align*}
because  the reduced loop of $\overline{p}_{(a,\partial_0b_n)}  \, b_n  \, p_{(a,\partial_1b_n)} \,   \cdots \overline{p}_{(a,\partial_0b_2)}  \, b_2  \, p_{(a,\partial_1b_2)}\, 
 \overline{p}_{(a,\partial_0b_1)}  \, b_1  \, p_{(a,\partial_1b_1)} $ is $p$ since 
$\partial_0 b_i =  \partial_1 b_{i+1}$. From this and the property (b) of ${\lambda}$,  causality  for $\ru^{\lambda}_{\cP}$ follows. 
The covariance of $(\ru^{\lambda}_\cP,U)$ follows by the one of $\lambda$ and of $\cP$ through a direct computation.\\
\indent Conversely, let $(\ru,U)$ be a causal and covariant connection system. Define 
\begin{equation}
\label{0Cc:8}
 {\lambda}^\ru(p_n\cdots p_2  p_1):= \ru_{a_n}(p_n)\cdots \ru_{a_2}(p_2)\ru_{a_1}(p_1) \ , \qquad p_i:a_i\to a_i \ , i=1,\ldots, n \ . 
\end{equation}
By the properties of a connection system, it is easily seen that $({\lambda}^\ru,U)$ is a causal and covariant representation of $\rL(\bR^4)$. Notice also that 
\[
 {\ru}^{\lambda^\ru}_{\cP_a}(b) 
 \ = \
 {\lambda}^{\ru}(\overline{p}_{(a,\partial_0b)}  \, b  \, p_{(a,\partial_1b)}) = \ru_{a}(\overline{p}_{(a,\partial_0b)}) \, \ru_a(b) \,  \ru_{a}(p_{(a,\partial_1b)}) \ . 
\]
So, if we set $\rt_a(a'):=\ru_a(p_{(a,a')})$ for any $a'\in\Si^a_0$, we have that 
$\rt_a(\partial_0b) \ru^{{\lambda}^\ru}_a(b)  =  \ru_a(b) \rt_a(\partial_1b)$ 
and 
\[
\ad_{U(Q)}(\rt_a(a'))=  \ad_{U(Q)}(\ru_a(p_{(a,a')}))
\ = \
\ru_{Qa}( p_{Qa,Qa'}) 
\ = \ 
\rt_{Qa}(Qa') \ ,
\]
for any $Q \in  \sP^\uparrow_+$.
Hence $\ru$ and  $\ru_\cP^{\lambda^\ru}$ are equivalent. 
Conversely, given a representation of the group of loops $\lambda$ and a loop $p : a \to a$, $p=b_n\ldots b_1$, we have 
\[
{\lambda}^{\ru_\cP^{\lambda}}(p)=\ru^{{\lambda}}_{\cP_a}(p) =  {\lambda}\left(\overline{p}_{(\partial_0b_n,a)}  \, b_n  \, p_{(\partial_1b_n,a)} \,   \cdots \overline{p}_{(\partial_0b_2,a)}  \, b_2  \, p_{(\partial_1b_2,a)}\, 
  \overline{p}_{(\partial_0b_1,a)}  \, b_1  \, p_{(\partial_1b_1,a)} \right)  = {\lambda}(p) \ .  
\]
So we have, up to equivalence, a 1-1 correspondence. 
\end{proof}
%
It is worth pointing out that (\ref{0Cc:7}) can be seen as the combinatorial counterpart 
of the analytical procedure outlined in Subsection \ref{0Bb}, which associates the primitive 1-form to a closed 2-form.  We shall see in Subsection \ref{0Ce} that these two procedures agree.

\paragraph{Gauge transformations.} Following \cite{CRV}, gauge transformations are now introduced as transformations sending a connection system to an equivalent connection system. More precisely, a \emph{gauge transformation} of a causal and covariant connection system $(\ru,U)$ is a family 
\begin{equation}
\label{eq.gauge}
\rg \ := \ \{ \rg_a:\Si^a_0\to \cU\cH  \ \ , \ \ a \in \Si_0 \} \ ,
\end{equation}
such that 
\begin{itemize}
\item[(a)] $\ad_{U(P)} (\rg_{a}(a'))= \rg_{Pa}(Pa')$ for any $a'\in\Si^a_0$ and $P\in \sP^\uparrow_+$;
\item[(b)] $\ad_{\rg_a(a)}(\ru(\cA_{\tilde a})) = \ru(\cA_{\tilde a})$ for any $\tilde a \supseteq a$,
\end{itemize}
where $\ru(\cA_{\tilde a})$ is the image of $\cA_{\tilde a}$ under the representation defined by $\ru$. 
Gauge transformations form a group $\cG^\ru$ under the pointwise multiplication 
$(\rg \cdot\tilde \rg)_a(a'):=  \rg_a(a') \, \tilde \rg_a(a')$. 
We call $\cG^\ru$ the group of \emph{gauge transformations} of $(\ru,U)$. \smallskip

Let now $\rg\in\cG^\ru$. Then the pair $(\ru^\rg,U)$, where 
\begin{equation}
\label{0Cc:9}
\ru^\rg_a(b) \ := \ g_a(\partial_0b)\, \ru_a(b)\,  g_a(\partial_1b)^* \ , \qquad b\in\Si_1^a \ , 
\end{equation}
is a causal and covariant connection system equivalent to $(\ru,U)$. 
In fact covariance is clear whilst, concerning causality, if $p:a\to a$ is a loop then  
\[
\ru^\rg_a(p) \ = \
\ru^\rg_a(b_n)\cdots  \ru^\rg_a(b_2)\,\ru^\rg_a(b_1) \ = \ 
\rg_a(a) \ru_a(p) \rg_{a}(a)^* \ .  
\]
So by property (b) of a gauge transformation we have that $\ru^\rg_a$ is also causal.\smallskip

It is worth pointing out that the degree of freedom in choosing different path-frame systems results to be a gauge transformation. 
In fact, let $(\lambda,U)$ be a causal and covariant representation of $\rL(\bR^4)$ and let $\cP$ and $\cQ$ be covariant path-frame systems.  
Moreover, let 
\[ 
\ru^{\lambda}_{\cP_a}(b)= {\lambda}\left(\overline{p}_{(a,\partial_0b)}  \, b  \, p_{(a,\partial_1b)} \right) \  \  ,  \ \ 
\ru^{\lambda}_{\cQ_a}(b)= {\lambda}\left(\overline{q}_{(a,\partial_0b)}  \, b  \, q_{(a,\partial_1b)} \right) 
\]
be the connection systems associated, \emph{via} $\lambda$, to the path-frame systems $\cP$ and $\cQ$. 
Then, defining 
 \begin{equation}
 \label{0Cc:10}
 \rg_a(a'):= \lambda(q_{(a,a')} \,\overline{p_{(a,a')}}) \ , \qquad a'\in\Si^a_0 \ , 
 \end{equation}
it is easily seen that $\rg$ is a gauge transformation of  $(\ru^\lambda_\cP,U)$ such that
$\ru^{\lambda,\rg}_\cP= \ru^\lambda_\cQ$. 
Notice, in particular, that $\rg_a(a)= 1$ for any $a \in \Si_0$, so the action (b) on $\ru(\cA_{\tilde a})$ is trivial.

\paragraph{Causal and covariant 2-cochains.}  We now give an equivalent description of representations of the net of causal loops on $\bR^4$ in terms of 2-cochains.
The idea is to observe that any connected component $\Si_*^a \subset \Si_*$, $a \in \Si_0$, is contractible. So we shall see, \emph{via} the cone construction,
that any loop admits a natural ``triangulation" in terms of 2-simplices.  As anticipated this is a key result since, as we shall see in the next section, it allows 
the electromagnetic field to induce representations of the net of causal loops.\\ 
\indent From now on it will be useful  to switch from the description of simplices  in terms of vertices and faces to the parametric 
description,  and conversely. To begin with, let us introduce the following 
\begin{definition}
\label{0Cc:14}
A \textbf{causal and covariant 2-cochain} is a pair $(\rw,U)$, where 
$U :  \sP^\uparrow_+ \to \cU\cH$ 
is a unitary representation and 
$\rw:\Si_2\to \cU\cH$ 
is a 2-cochain satisfying the properties 
\begin{itemize}
\item[(a)] $\rw(c)^*=\rw(\bar c)$ and $\rw(c)= 1$ if $c$ is degenerated;
\item[(b)] $[\rw(c_1),\rw(c_2)]=0$ if $|c_1|\perp |c_2|$;
\item[(c)] $\ad_{U(P)} \circ \rw = \rw\circ P$, for any $P \in \sP^\uparrow_+$,
\end{itemize}
where in $(b)$ the symbol $|\cdot |$ refers to the support of a 2-simplex, equation (\ref{0A:14}).
\end{definition}
We need a preliminary observation in order to prove that such 2-cochains provide an equivalent description  of the representations of the net of causal loops.\\ 
\indent Let $a=(a_0,f) \in \Si_0$. For any $b=(r_b,f)\in\Si_1^a$, written in parametric form, we consider the smearing 2-simplex $h^{a}(b)\in\Si^a_2$ defined by the cone construction (\ref{0A:8a}). Namely     
\begin{equation}
\label{0Cc:14a}
 h^a(b):= (h^{a_0}(r_b), f)\ \   ,  \qquad b=(r_b,f)\in \Si^a_1\ .  
\end{equation}
Note that $h^a(b)$ is characterized 
as the unique smearing 2-simplex having  faces
\begin{equation}
\label{0Cc:14b}
\partial_{0} h^a(b)= b \ \ ,  \ \ \partial_{2} h^a(b)= e_{(\partial_1b,a)} \ \ , \ \ \partial_{1} h^a(b) = e_{(\partial_0b,a)} \ .  
\end{equation}
Equivalently, this is the unique 2-simplex whose path-boundary is the loop 
$ \boldsymbol{\partial} h^a(b)= $ $\overline e_{(\partial_0b,a)} \, b \, e_{(\partial_1b,a)}$. 
This implies that  $h^a(\bar b)=  \overline{h^a (b)}$ for any $b\in\Si_1^a$,  since     
$ \boldsymbol{\partial}  h^a(\bar b)  = 
\overline e_{(\partial_1b,a)} \, \bar b \, e_{(\partial_0b,a)}$
$ = \overline{  \boldsymbol{\partial}  h^a (b)}$.\smallskip

On these grounds, given a causal and covariant 2-cochain $(\rw,U)$, define  
\begin{equation}
\label{0Cc:15}
 {\lambda}^\rw(p) \ := \ \rw(h^a(b_n))\cdots \rw(h^a(b_1))) \ , \qquad p=b_n\cdots b_1:a\to a  \ .
\end{equation}
This definition is well posed because, as observed before, $h^a(\bar b)= \overline{h^a(b)}$, so ${\lambda}^\rw(p)$ is independent of the reduction of the loop $p$.
For the same reason we have
\[
\lambda^\rw(\bar p) \ = \ \lambda^\rw(p)^*
\ \ , \ \ 
\lambda^\rw(pq) = \lambda^\rw(p) \lambda^\rw(q)
\ \ , \ \ p,q : a \to a
\ ,
\]
and this implies that $\lambda^\rw$ is a representation of the group of loops.

Given a loop $p : a \to a$, $p= b_n\cdots b_1$, we can assume without loss of generality that it is reduced. 
If $|b_i|\subseteq o$ for any $i$, then $|h^a(b_i)|\subseteq o$ for any $i$ since double cones are convex, and this implies that $\lambda^\rw$ is causal. 
Finally, covariance follows from the transformation properties of $\Si_*$, and we conclude that $(\lambda^\rw,U)$ is causal and covariant.\smallskip

Conversely, it is easily seen that  if $({\lambda},U)$ is a causal and covariant representation of $\rL(\bR^4)$, then the pair 
$(\rw^{\lambda},U)$, where 
\begin{equation}
\label{0Cc:16}
 \rw^{\lambda}(c) \ := \ {\lambda}(\boldsymbol{\partial} c) \ ,\qquad c\in\Si_2 ,  
\end{equation}
is a causal and covariant 2-cochain.  This correspondence is 1-1: in fact, if $c$ is a 2-smearing simplex  with path-boundary 
$\boldsymbol{\partial}c:a\to a$, then 
\[
 \rw^{{\lambda}^\rw}(c)= {\lambda}^\rw(\boldsymbol{\partial} c) = \rw(h^a(\overline{\partial_1c}))\, \rw( h^a(\partial_0c))\, \rw(h^a(\partial_2c))=  \rw(h^a(\partial_0c))= \rw(c) \ , 
\]
because, as can be easily seen by  (\ref{0Cc:14b}),  $h^a(\overline{\partial_1c})$ and $h^a(\partial_2c)$ are degenerated 2-simplices ($a$ being a vertex of both $\partial_1c$ and $\partial_2c$)
and $h^a(\partial_0c)=c$.  On the other hand, for any loop $p : a \to a$, $p=b_n\cdots b_1$, we have,
using a calculation similar to that used in Lemma \ref{0Cc:6}, 
\[
{\lambda}^{\rw^{\lambda}}(p) \ = \ 
\rw^{\lambda}(h^a(b_n))\cdots \rw^{\lambda}(h^a(b_1))) \ = \
{\lambda}(\boldsymbol{\partial} h^a(b_n))\cdots  {\lambda}(\boldsymbol{\partial} h^a(b_1)) \ = \ 
{\lambda}(p) \ .
\]
The results of this section are summed up in the following theorem. 
\begin{theorem}
\label{0Cc:17}
There exists, up to equivalence, a 1-1 correspondence between the following: 
\begin{itemize} 
\item[(i)] Covariant representations $(\pi,U)$ of the net of causal loops $(\cA,\alpha)_K$; 
\item[(ii)] Causal and covariant representations $({\lambda},U)$ of the group of loops $\rL(\bR^4)$; 
\item[(iii)] Causal and covariant connection systems $(\ru,U)$; 
\item[(iv)] Causal and covariant 2-cochains $(\rw,U)$.
\end{itemize} 
\end{theorem}

\subsection{The electromagnetic field representation} 
\label{0Ce}

We now show that the electromagnetic field induces a representation of the net of causal loops $(\cA,\alpha)_K$ in terms of  a causal and covariant 2-cochain $(\rw^{em}, U)$. This yields the connection system $\ru^{em}$ with the relative group of gauge transformations, and we study the relation between them and the  analytical procedure associating to $F_{\mu\nu}$ the potential system $A_\mu$ with the corresponding gauge transformations. The result is that we find a complete coherence: the exponential of the line integral of $A_\mu$ defines a connection system $\ru^{pot}$ which turns out to be gauge-equivalent to $\ru^{em}$, and gauge trasformations of $A_\mu$ define gauge transformations of $\ru^{em}$.\smallskip

\paragraph{The electromagnetic 2-cochain.} We now show that the electromagnetic field defines a causal and covariant 2-cochain, hence a covariant representation of the net of causal loops. To this end  we introduce the following ``homological" deformation of a 2-simplex.\smallskip

Let $c=(\si_c,f)\in\Si^a_2$ be a 2-simplex written in 
the parametric form, where $\si_c$ is the triangular surface  associated with the vertices  of $c$. 
Given  $a=(a_0,f)$,  a $0$-simplex having the same smearing function as $c$,
let  $ h^a(c)=(h^{a_0}(\si_c),f)$ be the  3-simplex obtained via the cone construction. 
According to the definition (\ref{0A:8a}) we have
\[
\partial_0h^a(c) = c \ , \ 
\partial_1h^a(c) =  h^a(\partial_0c) \ , \ 
\partial_2h^a(c) =  h^a(\partial_1c) \ , \ 
\partial_3h^a(c) =  h^a(\partial_2c) \ ,
\]
and the smearing 2-chain 
\begin{equation}
\label{0Cc:19} 
 h^a(\partial c) \ = \ h^a(\partial_0 c) - h^a(\partial_1 c) + h^a(\partial_2 c)   \ , \qquad c\in\Si^a_2 \  , 
\end{equation}
has boundary $\partial c$ by (\ref{eq.h}). \smallskip

Now, let us consider the electromagnetic field $F_{\mu\nu}$ and the corresponding unitary representation $U :  \sP^\uparrow_+ \to \cU\cH$ defined in Section \ref{0A}. Using (\ref{0B:4}), we define 
\begin{equation}
\label{0Cc:21}
 \rw^{em}(c) \ := \ \exp(i F\left< \si_c,f\right>) \, \in \cU\cH \ \ , \ \qquad c\in\Si_2 \ . 
\end{equation}
\begin{theorem}
\label{0Cc:22}
The pair $(\rw^{em},U)$ is a causal and covariant 2-cochain, defining a covariant representation of the net of causal loops $(\cA,\alpha)_{K}$ and fulfilling the following causality property. Given  $c\in \Si_2^a$,  $c' \in \Si_2^{a'}$ with $a=(a_0,f),a'=(a'_0,f')\in\Si_0$,  we have 
\begin{equation}
\label{eq.link}
| h^a(\partial c)|  \perp  | h^{a'}(\partial c')| \ \ \Rightarrow \ \ [\rw^{em}(c),\rw^{em}(c')]  =  0  \ ,
\end{equation}
where $|h^a(\partial c)|$, $|h^{a'}(\partial c')|$ denote the supports of the smearing 2-chains $h^a(\partial c)$, $h^{a'}(\partial c')$ respectively, see  (\ref{0A:15}). 
\end{theorem}
\begin{proof}
According to the Definition \ref{0Cc:14}, the first part of the statement follows from Lemma \ref{0B:5} and Theorem \ref{0Cc:17}.  
Concerning the second statement, we observe that, given   $c=(\si_c,f)\in\Si^a_2$,  the surface $\si_c$ and the  2-chain 
$h^{a_0}(\partial c)$  have the same boundary $\partial c$.  So applying Lemma \ref{lem.0B:6}  we have  
$\rw^{em}(c) \ = \ \exp(i F\left<h^{a_0}(\partial c),f\right>)$,  
and the proof follows by Lemma \ref{0B:5}$(iii)$. 
\end{proof}
The latter property of $\rw^{em}$ stated in the previous theorem is a refinement of the causality used in \cite{CRV} and in Theorem \ref{thm.0B:6}, in fact it does not involve double cones containing $|c| , |c'|$. It is easily seen that (\ref{eq.link}) applies to 2-simplices $c,c'$ whose boundaries are causally disjoint and  form a trivial link, 
so the previous theorem says that $\rw^{em}(c)$ and $\rw^{em}(c')$ commute in accord with the considerations in \cite[Section 1]{Rob77}.

\paragraph{Connection systems and gauge transformations.} We have already pointed out that the combinatorial procedure used to extract the connection system from a 2-cochain is analogous to the analytical one, used to define the primitive 1-form starting from a closed 2-form. Now our aim is to prove that these two procedures are, up to equivalence, the same.
To be precise, we shall compare the connection system associated with $\rw^{em}$ with the one defined by the exponential of the line integral of $A_\mu$: we shall prove that these are the same, up to a gauge transformation in the sense of (\ref{eq.gauge}). \smallskip

Consider the causal and covariant representation $(\lambda^{em},U)$ of the group of loops associated with $(F_{\mu\nu},U)$. This is  defined, according to  (\ref{0Cc:15}), by 
\begin{equation}
\label{0Cc:23}
\lambda^{em}(p) \ = \ \rw^{em}(h^a(b_n))\cdots \rw^{em}(h^a(b_1)) \ , \qquad  p=b_n\cdots b_1:a\to a.     
\end{equation}

\begin{remark}
It is worth observing that Theorem \ref{0Cc:22} says that $\lambda^{em}$ is \emph{localized on loops which are path-boundaries of 2-simplices}:  
given two loops $\boldsymbol{\partial}c$ and  $\boldsymbol{{\partial}}c'$, $c,c'\in\Si_2$, having causally disjoint supports and forming a trivial link, 
from Theorem \ref{0Cc:22} we obtain 
\begin{equation}
\label{0Cc:24}
[\lambda^{em}(\boldsymbol{\partial} c),\lambda^{em}(\boldsymbol{\partial} c')] =0 \   . 
\end{equation}
In fact, since the loops are not linked together, it is possible to find $a$ and $a'$ as in the recalled theorem and  $\lambda^{em}(\boldsymbol{\partial} c)=\rw^{em}(c)$ by definition.  
Notice that in the case of  free fields  the relation  (\ref{0Cc:24}) should also hold for arbitrary loops,  because the composition of two Weyl operators gives the sum of the exponents up to a phase factor. So, if $p$ is the boundary of a 2-surface, then $\lambda^{em}(p)$ corresponds to the exponential of the integral of the fields on this surface up to some phase factor. 
Since phase factors do not affect commutativity,  (\ref{0Cc:24}) should hold. 
\end{remark}

Let $\cP$ denote a covariant path-frame system. We consider the causal and covariant connection system $(\ru^{em}_\cP,U)$ associated with $(\lambda^{em},U)$ which is defined, according to (\ref{0Cc:7}), by   
\begin{equation}
\label{0Cc:25}
  \ru^{em}_{\cP_a}(b) \ = \ \lambda^{em}(\overline{p}_{(a,\partial_0b)} \, b \, p_{(a,\partial_1b)})
  \ , \qquad 
  a \in \Si_0 \ , \ b\in\Si_1^a \ . 
\end{equation}
As observed in Subsection \ref{Cc},  a changing of the path-frame system leads to an equivalent connection system by means of a gauge transformation. In particular, considering the Euclidean path-frame system $\cE$ and a generic one $\cP$ we have 
\begin{equation}
\label{0Cc:25a}
\rg_a(\partial_0b)\,  \ru^{em}_{\cE_a}(b) \ = \ \ru^{em}_{\cP_a}(b) \, \rg_a(\partial_1b) \ , \qquad b\in \Si_1^a \ ,  
\end{equation}
where 
\begin{equation}
\label{0Cc:25b}
 \rg_a(a') \ := \ \lambda^{em}\left(p_{(a,a')} \overline{e}_{(a,a')}\right) \ , \qquad  a'\in\Si^a_0 \ .
\end{equation}
Concerning the connection system $\ru^{em}_{\cE}$,   since $h^a(e_{(a,\partial_i b)})$, $i=0,1$, is a degenerated 2-simplex,  
applying (\ref{0Cc:23}) and (\ref{0Cc:25}) we find
$\lambda^{em}(\overline{e}_{(a,\partial_0b)} \, b \, e_{(a,\partial_1b)}) = \rw^{em}(h^a(b))$  and this implies
\begin{equation}
\label{eq.uF}
\ru^{em}_{\cE_a}(b) \ = \ \exp(i F\left< h^{a_0}(r_b),f\right>) \ , \qquad b\in\Si^a_1 \ . 
\end{equation}
We now construct a connection system starting from the potential system $A_\mu$. 
We begin by  defining  a  family $\ru^{pot}:=\{\ru^{pot}_a : \Si_1^a \to \cU\cH\}$  of unitary 1-cochains, for $a = (a_0;f)\in\Si_0$, as
\begin{equation}
\ru^{pot}_a(b) \ := \ \exp \left(i A^{a_0}\left<r_b,f\right> \right) \ , \qquad   b=(r_b,f)\in\Si_1^a \  . 
\end{equation}
It is possible to check directly that $(\ru^{pot},U)$ is a causal and covariant connection system, but we do not need to do that
because $\ru^{pot} = \ru^{em}_\cE$. In fact, Proposition \ref{0Bb:16} implies
$A^{a_0}\left<r_b,f\right> = F\left<h^{a_0}(r_b), f\right>$,
so by (\ref{eq.uF}) we have $\ru^{em}_\cE  =  \ru^{pot}$ as desired.
In conclusion, the equation \ref{0Cc:25a} gives:
\begin{theorem}
Let $(F_{\mu\nu},U)$ denote the quantum electromagnetic field.
Under the above notations, for any covariant path-frame system $\cP$ we have 
\begin{equation}
(\ru^{em}_{\cP},U) \ \cong \ (\ru^{em}_{\cE},U)  \ = \  (\ru^{pot},U) \ , 
\end{equation}
where $\cE$ is the Euclidean path-frame system and the equivalence $\cong$ is realized by the gauge transformation (\ref{0Cc:25b}).  
\end{theorem}
In words, the abstract connection system $\ru^{em}_\cP$  coincides with the one obtained as the exponential of the line integral of $A_\mu$, 
up to gauge transformations.
\begin{remark}
Two observation are in order. 
\begin{enumerate}
\item If $\rg = \{ \rg^z \}_{z \in \bR^4} \in \cG$ is a gauge transformation of $A_\mu$ in the sense of  Subsection \ref{0Bc}, then setting
\[
\rg^\ru_a(a') \ := \ \exp ( i \rg^{a_0}(a'_0,f) ) \ , \qquad a'=(a'_0,f) \in \Si_0^a \ ,
\]
we easily find that $\rg^\ru$ is a gauge transformation of $\ru^{pot}$ leaving any local algebra $\cA_o$, $o \in K$, pointwise fixed.
\item Let $a = (a_0;f) \in \Si_1$ and $b,b' \in \Si_1^a$ with $\partial_1b_1 = \partial_0b'$; then the composition
of the corresponding lines $r_b$ and $r_{b'}$ yields  the curve $r_b*r_{b'}$.
It is worth to stress that, also in the simpler case of \emph{free} electromagnetic field, the operator 
$\exp\left(iA^{a_0}\left<r_b*r_{b'},f\right> \right)$ is different from the one obtained by the combinatorial product of the two single paths,
\[
\ru^{pot}(b b') \ = \ 
\ru^{pot}(b) \ru^{pot}(b') \ = \ 
\exp \left(i A^{a_0}\left<r_b,f\right> \right) \cdot \exp \left(i A^{a_0}\left<r_{b'},f\right> \right) \ .
\]
Hence, path composition is not preserved under the quantization process.
\end{enumerate}
\end{remark}

\section{Concluding remarks}
In a previous paper \cite{CRV} we used loops in a globally hyperbolic space-time to generate a causal and covariant net of $\rC^*$-algebras, called  the net of causal loops.
We presented some of its representations in terms of connection systems, i.e.\ families of ``abstract"  connections fulfilling causality and covariance as properties of the family and not of a single connection. 
Local gauge transformations were defined as maps between equivalent connection systems,  leaving  element-wise invariant the loop algebras.  
The connections are recovered  by the representations using a path-frame: a choice of paths joining any point of the space-time with a fixed point, the pole. 
Letting varying the pole yields  the connection system. \smallskip 

In this paper, starting  only from the (possibly charged) quantum electromagnetic field $F_{\mu\nu}$,
we reconstructed a potential 1-form $A^z_\mu$ with reference to the pole $z$,  i.e.\ the center of a contracting homotopy of the Minkowski space-time. Actually this homotopy  corresponds to the abstract path-frame used to define a connection.  According to different choices of the pole $z$, a potential system $A_\mu=\{A^z_\mu,\, z\in \bR^4 \}$ is obtained, and a series of outcomes that agree with the abstract formulation given by the authors in \cite{CRV} follows. 
In particular, $A_\mu$ gives  a  connection system and, in turns, a covariant representation of the net of causal loops. Furthermore,  \emph{local} gauge transformations defined in terms of the potential system $A_\mu$ coincide with the ones defined in terms of the corresponding connection system. 
\smallskip

For these results, an important outlook  is the comparison with the global conditions defining the charge classes for theories with long-range interaction developed in \cite{BdR}. 
In that paper a result of DR-duality type is obtained for the case of simple charge classes, giving a \emph{global} Abelian gauge group. In particular, motivated by sound physical reasons, the charge classes are defined over a time-like cone in the Minkowski space-time, using a family of its subsets called hypercones.
\smallskip

To better understand these dual aspects of long-range interactions and local gauge theory on loops, it seems to be useful to study the geometrical facets of the two cases, e.g.\ the choice of a time-like cone with its apex and hypercones on the one side, and of a path-frame with its pole on the other side. 

In this regard, a hopeful hang is the realization for QED  of an abstract connection in terms of the charge transporter proposed in equation (\ref{due}), in fact the presence of the charged massive field $\psi$ and of the massless potential form $A^z_\mu$ should make possible to shift charges to infinity.
Compatibly with the two frameworks, the pole $z$ may reveal the choice of a future time-like cone.  

\smallskip
Under a different perspective, it may be of interest to explore the relation between the approach of this paper and others, e.g.\ TQFT,
referring to the interpretation of observables localized on loops.

\end{document}